\documentclass[12pt]{article}
\usepackage[margin=1in]{geometry}
\usepackage[small]{titlesec}
\usepackage[utf8]{inputenc}
\usepackage{amsmath,amssymb,amsthm,graphicx,setspace,paralist,booktabs,rotating,subcaption,float}
\usepackage[titlenumbered,ruled]{algorithm2e}
\usepackage{times}
\usepackage{multirow}
\AtBeginDocument{%
	\abovedisplayskip=5pt plus 2pt minus 2pt
	\belowdisplayskip=\abovedisplayskip
	\abovedisplayshortskip=2pt plus 2pt minus 2pt
	\belowdisplayshortskip=\belowdisplayskip}
\titlespacing*{\section}{0pt}{*2}{*1}
\titlespacing*{\subsection}{0pt}{*2}{*1}

\usepackage[usenames,dvipsnames]{color}
\usepackage{tabu}
\usepackage{graphicx}
\usepackage{graphics}
\usepackage{caption}
\usepackage{subcaption}
\usepackage{bbm}
\usepackage{pstricks, pst-node}
\usepackage{tikz}
\usepackage[linewidth=1pt]{mdframed}
\usepackage{multirow}
\usepackage{natbib}
\usepackage{pgfplots}
\usepackage{hyperref}
\usepackage{bm}
\usepackage{xr}
\usepackage{float}

\pgfplotsset{compat=1.10}
\usepgfplotslibrary{fillbetween}
\usetikzlibrary{patterns}
\usetikzlibrary{shapes,arrows}

\newcommand{\R}{\mathcal{R}}

\newtheorem{Corollary}{Corollary}
\newtheorem{Proposition}{Proposition}

\newtheorem{Theorem}{Theorem}

\newtheorem{Assumption}{Assumption}

\newcommand{\supp}{{\rm supp}}

\newcommand{\cid}{\overset{d}{\to}}

\def\cI{\mathcal{I}}

\def\bbK{\bbK}

\def\cR{\mathcal{R}}
\def\cS{\mathcal{S}}

\def\1{{\mathbf 1}}

\def\bbK{\mathbb{K}}

\def\nmax{ n_{\max} }
\def\beq{\begin{equation}}
\def\eeq{\end{equation}}

\def\tr{^{\intercal}}
\newcommand{\norm}[1]{\left\lVert#1\right\rVert}

 \newcommand{\ignore}[1]{}

\DeclareMathOperator*{\argmin}{arg\,min} 
\def\sig{\sigma}
\def\gam{\gamma}
\def\Sig{\Sigma}
\def\eps{\varepsilon}
\newcommand{\Cov}{\mathrm{Cov}}
\newcommand{\Var}{\mathrm{Var}}

\newcommand{\E}{\mathrm{E}}

\newcommand{\bI}{\mathbbm{1}}
\newcommand{\I}{\cI}
\newcommand{\Iss}{\widehat I}

\newcommand{\fx}{f( \beta_0 + x_i \beta )}

\newcommand{\hfx}{f( \hat \beta_0 +   x_i \hat\beta  )}
\newcommand{\gx}{g( \gamma_0 + x_i \gamma)}

\newcommand{\hgx}{g( \hat \gamma_0 + x_i \hat\gamma)}
\newcommand{\hsfx}{f(\hat\beta_0 + x_i \hat\beta)}
\newcommand{\hsgx}{g(\hat \gamma_0 + x_i \hat\gamma)}

\newcommand{\prop}{\mathrm{prop}}
\DeclareMathOperator{\Tr}{Tr}

\newcommand{\blind}{1}

\allowdisplaybreaks[1]

\begin{document}

\def\spacingset#1{\renewcommand{\baselinestretch}%
{#1}\small\normalsize} \spacingset{1}


\date{}

\if1\blind
{
	\title{A Regression-based Approach to Robust Estimation and Inference for Genetic Covariance} 

 \author{Jianqiao Wang$^1$, Sai Li$^2$, Hongzhe Li$^3$
 	\hspace{.2cm}\\
 	$^1$Department of Biostatistics\\ Harvard University \\
    $^2$Institute of Statistics and Data Science,  \\
    Renmin University of China\\
    $^3$Department of Biostatistics, Epidemiology and Informatics \\ University of Pennsylvania\\
     }
  \maketitle
} \fi

\if0\blind
{
  \bigskip
  
  \begin{center}
    {\LARGE\bf A Regression-based Approach to Robust Estimation and Inference for Genetic Covariance 
    }
\end{center}
  \medskip
} \fi

\thispagestyle{empty}
\begin{abstract}
\setstretch{1.2}
Genome-wide association studies (GWAS) have identified thousands of genetic variants associated with  complex traits, and some variants are shown to be associated with multiple  complex traits.  Genetic covariance between two traits is defined as the underlying covariance of genetic effects and can be used to measure the shared genetic architecture.  The data used to estimate such a genetic covariance can be from the same group or different groups of individuals,  and the traits can be of different types or collected based on different study designs. This paper proposes a unified regression-based approach to robust  estimation and inference for genetic covariance of  general  traits that may be associated with genetic variants nonlinearly.   The asymptotic properties of the proposed estimator are provided  and are shown to be robust under certain  model mis-specification. Our method under  linear working models provides a robust inference for the narrow-sense genetic covariance, even when both linear models are mis-specified.   Numerical experiments are performed to support the theoretical results. Our method  is applied to an outbred mice GWAS data set to study the overlapping genetic effects between the behavioral and physiological phenotypes. The real data results  reveal interesting genetic covariance among different mice developmental traits. 

\end{abstract}

\noindent%
\vfill
\spacingset{1.695} 

\clearpage
\pagenumbering{arabic} 
\section{Introduction}
\label{sec:intro}

Genome-wide association studies (GWAS) have identified thousands of variants related to complex traits. Some  complex traits are shown to have a shared genetic etiology, including various autoimmune diseases \citep{YunLi} and psychiatric disorders \citep{cross2019genomic}. Studying the shared genetic architecture and investigating the relationship between the genetics of different traits can provide insights into the underlying biological mechanism. 
There has been a great  interest in quantifying the overlapping genetic effects between pairs of traits based on GWAS data. With individual-level  genotype data,  genetic effects on the traits can be  modeled by some functions of the observed genetic variants such as linear functions, and the single nucleotide polymorphism (SNP)-based estimate of genetic covariance  can be  derived \citep{lee2012estimation,bulik2015atlas} based on the assumed models. 

Although estimating the shared  genetic effects has received much attention recently, several issues remain to be addressed. First, most existing literature focuses on  traits that are generated from linear models. For binary traits, such as the occurrences of diseases, linear models are not suitable. However, there are few approaches with theoretical guarantees for dealing with nonlinear trait models. 
Second, the true underlying models for the traits are unknown in practice and need to be specified. When the  models are mis-specified, the estimated genetic effects and the corresponding genetic covariance can be biased. For estimating the shared  genetic effects between two traits, model mis-specification is more likely to happen. Thus, it is important to develop robust methods for estimating the genetic covariance. 
Third, the genetic  data sets can be collected based on different designs, further complicating the analysis of genetic covariance.  For example, two traits and the corresponding genotypes may be collected on  the same samples or independent samples. In a more challenging scenario, one trait may be collected from a cohort study, while the other may be from a case-control study. 
As far as we know, existing methods cannot be directly applied to such data or do not have any statistical guarantee.  

Motivated by the challenges mentioned above, we study estimation and inference for genetic covariance based on the individual-level GWAS data. We consider generic population models for the trait that is possibly associated with the genotypes nonlinearly. Specifically,  we consider the following models for the two traits, 
\begin{equation}
\label{true-model}
y_i= f^{*}(x_i)+ \eps^{*}_i, \quad z_i=g^{*}(x_i)+ v^{*}_i, 
\end{equation}
where $y_i$ and $z_i$ are two trait values, including  continuous complex trait measures,  disease outcomes, or gene expressions, and $x_i\tr = (x_{i1}, \cdots, x_{ip} ) \tr \in\R^{p}$ denotes the $i$-th observation of $p$ genetic variants, each coded as 0, 1 and 2.  In typical GWAS,  $p$ can be larger and much larger than the sample sizes. The functions $f^*(\cdot)$ and $g^*(\cdot)$ denote the true conditional mean functions of $y_i$ and $z_i$, respectively,  such that $\E[\eps^{*}_i\mid x_i]=\E[v^*_i\mid x_i]=0$. In genetic terms, $f^{*}(x_i)$ is the genetic effect  of $y_i$, representing the total effects of all the  genetic variants considered.  
The error terms $\varepsilon^{*}_i$ and $ v^{*}_i $ can be dependent due to possible confounding effects. 
The functions $f^*(\cdot)$ and $g^*(\cdot)$ can be known or unknown.  We emphasize that Model \eqref{true-model} includes a general class of trait models, which can be parametric or non-parametric. 

First considered in \cite{searle1961phenotypic}, the genetic covariance of traits $y_i$ and $z_i$ is defined as the covariance of their conditional mean functions $f^*(x_i)$ and $g^*(x_i)$ ,
\begin{equation}
\label{def:true-para}
I^{*} = \Cov(\E[y_i|x_i], \E[z_i|x_i]) = \Cov(f^{*}(x_i), g^{*}(x_i)).
\end{equation} 
Invoking the law of total variance formula, parameter $I^*$ represents the covariance of two traits explained by genetic variants. When the traits are binary,  this genetic covariance is defined at the observed scale instead of the liability scale as commonly used in probit-mixed effect models. 
A closely related concept is heritability. The heritability of trait $y_i$ is defined as $\Var(f^{*}(x_i))/ \Var(y_i)$, which is the proportion of phenotypic variance explained by genetic variants. 

While $I^*$ captures the covariance due to total genetic effects, sometimes the contribution of additive genetic effects can be of interest, especially for continuous complex traits. Analogous to the narrow-sense heritability \citep{tenesa2013heritability}, the narrow-sense genetic covariance measures the additive genetic covariance by fitting linear models to the traits.
Specifically, narrow-sense genetic covariance is defined as the bilinear functional $\beta\tr \Sigma \gamma$, where $\beta$ and $\gamma$ are the regression coefficients in the linear models representing the additive effects, $f^*(x_i)=\beta_0+x_i\beta$ and  $g^*(x_i)=\gamma_0+x_i\gamma$, and $\Sig$ is the covariance of $x_i$. We mention that even if $f^*(\cdot)$ or $g^*(\cdot)$ are nonlinear functions, one can still define $\beta$ and $\gamma$ through working linear models. That is, even if the linear models are wrongly specified, the definition of narrow-sense genetic covariance is still valid but $I^*\neq \beta\tr \Sigma \gamma$ in general. Formal definitions of $\beta$ and $\gamma$ are given in (\ref{def:def-beta}).
As far as we know, the estimation and inference for $\beta\tr \Sigma \gamma$ have not been studied when the linear models are  mis-specified. It is unknown whether existing methods are valid for the inference on $\beta\tr \Sigma \gamma$ under model mis-specification.

\subsection{Related literature}
Many methods have been proposed for estimating heritability and genetic covariance in genetics literature. Among them, the linear mixed-effects model is one of the most popular choices and is adopted in the GCTA-GREML method  \citep{lee2013genetic} and  linkage disequilibrium score regression (LDSC) method \citep{bulik2015atlas}. In a linear mixed-effects model, the effects of genetic variants are assumed to be \textit{i.i.d.} normally distributed and the genetic covariance equals the covariance between the random effects with the genotype data standardized to have a unit variance. 
A closely related model is the latent liability model \citep{lee2013genetic} for binary traits, 
which  can be viewed as a random-effects probit model.  
Despite their  popularity, the assumptions of liability model or linear random-effects models have been questioned and  their estimates are not robust to the distributional assumptions on the effects \citep{speed2017reevaluation, WangLi}.   Some partition-based methods have been proposed to improve the heritability estimation \citep{yang2015genetic, gazal2017linkage}, where  more complicated random-effects models are used to make the estimation accurate \citep{gazal2019reconciling, evans2018comparison}. 

The estimation and inference for heritability and genetic relatedness in fixed-effects linear trait models have been studied in statistics literature.  Assuming correctly specified high-dimensional linear models, \cite{cai2018semi} studies estimation and inference for heritability based on the penalized regressions. 
\cite{verzelen2018adaptive} proposes an adaptive procedure to estimate heritability by combining the sparsity-based method and the method of moments in linear models. \cite{guo2017optimal} studies the estimation of genetic relatedness $\beta\tr \gamma$ between two continuous traits in correctly-specified high-dimensional linear models, which is closely related to the genetic covariance.  They assume that the noises in the two traits models are independent, but do not provide methods for inference.

\subsection{Our contributions}
As we reviewed above,  most of existing methods for estimating the genetic covariance  are developed  for continuous traits assuming correctly specified linear models and the independence of the errors of the  two traits. A robust  inference procedure is needed  when any of the model assumptions are violated. 
Our paper aims to fill this gap. 
We propose a regression-based approach to the estimation and inference of genetic covariance $I^*$ defined as equation \eqref{def:true-para} for the traits that can be nonlinearly  associated with  high-dimensional genetic variants.   Based on sparse  high-dimensional generalized linear models (GLMs) as working models, we propose an estimator and construct the confidence interval for genetic covariance. Our method makes use of all the available observations, allowing  samples for the two traits to overlap.   Our approach can also be used to estimate and make inferences for the narrow-sense genetic covariance  $\beta\tr \Sigma \gamma$ when linear working models are assumed.

We  show  the robustness  of the proposed estimator of  genetic covariance. 
Firstly,  when both working models are correctly specified GLMs, the proposed method provides an asymptotically normal estimator for $I^*$ under proper conditions.
Secondly,  when the models are possibly mis-specified, the proposed estimator is still consistent as long as one conditional mean model is correctly specified. If one linear working model is used and is correctly specified, then the asymptotic normality of our estimator of   $I^*$ still holds under certain sparsity conditions.



The rest of the paper is organized as follows. Section  \ref{sec:method} introduces a unified method of estimation and inference  for genetic covariance and narrow-sense genetic covariance. Section \ref{sec:theory} establishes the theoretical properties of our estimator  for the genetic covariance under  correct or mis-specified models, including  the method for estimating the narrow-sense generic covariance. 
In Section \ref{sec:simulation}, the performance of the estimator is evaluated in different settings using numerical experiments. In Section \ref{sec:real-data}, the proposed method is applied to an outbred mice data set to study the genetic covariance between the behavioral and physiological traits. A discussion is provided in Section \ref{sec:discuss}. The proofs of main theorems and the extended simulation studies are given in the Supplementary Materials.

\paragraph{Notations} Given a symmetric matrix $A$, we use $\norm{A}_{2}$ to denote the $\ell_2$-operator norm and \(\lambda_{\max }({A})\) and \(\lambda_{\min }({A})\) respectively to denote the largest and the smallest eigenvalue of $A$. For a design matrix $X\in\R^{n\times p},$ let $x_i$ denote the $i$-th row of $X$, and $x_{ij}$ denote the $j$-th element of the row vector $x_i.$ For an index set $\I$, we denote its cardinality as $|\cI|$. For a vector $u$, its support is represented by $\supp(u).$ For two positive sequences $a_n$ and $b_n$, $a_n \lesssim b_n$ means $a_n \leq Cb_n$ for sufficiently large $n$, \(a_{n} \gtrsim b_{n}\) if \(b_{n} \lesssim a_{n}\) and \(a_{n} \asymp b_{n}\) if $a_n \lesssim b_n$ and $b_n \lesssim a_n$. We write $b_{n} \ll a_{n}$ if $\lim \sup b_n/a_n \to 0.$ \(c\) and \(C\) are used to denote generic positive constants that may vary from place to place. 
\ignore{For a random variable \(U,\) its sub-Gaussian norm is defined as \(\|U\|_{\psi_{2}}=\sup _{q \geq 1} \frac{1}{\sqrt{q}}\left(\E|U|^{q}\right)^{\frac{1}{q}}.\) For a random vector \(U \in {\cR}^{p},\) its sub-Gaussian norm is defined as \(\|U\|_{\psi_{2}}=\sup _{v \in \cS^{p-1}}\|\langle v, U\rangle\|_{\psi_{2}}\)  where \(\cS^{p-1}\) is the unit sphere in \({\cR}^{p} .\) For a sequence of random variables $U_n$ indexed by $n,$ we use $U_n \cid N(0,1)$ to represent that $U_n$ converges to standard normal distribution.}

\section{Estimation of genetic covariance and construction of confidence interval}
\label{sec:method}

\subsection{Fitting the working models}
\label{sec:work-model}
To estimate the genetic covariance, we consider fitting potentially mis-specified GLMs to the two traits \citep{nelder1972generalized}. The fitted models are treated as working models.  In GLMs with canonical link functions, the regression coefficients $\beta$ and $\gamma$ are defined as the minimizer of the population negative log-likelihood functions, i.e.,
\begin{align}
&(\beta_0, \beta) = \argmin_{ \beta_0 \in \cR, \, \beta \in \cR^{p}} \E \left[ F(\beta_0 + x_i \beta) - y_i (\beta_0 + x_i \beta ) \right]\nonumber\\
& (\gamma_0,\gamma) = \argmin_{ \gamma_0 \in \cR,\, \gamma \in \cR^{p}} \E [ G(\gamma_0 + x_i \gamma) - z_i (\gamma_0 + x_i \gamma) ].\label{def:def-beta}
\end{align}
Functions  $F(t)$ and $G(t)$ have known parametric forms, for example, $F(t) = t^2/2$ gives a linear model and  $F(t) =  \log(1 + \exp(t))$ gives a logistic model. For differentiable $F(t)$ and $G(t)$, $f(t) = F^{\prime}(t)$ and $g(t) = G^\prime(t)$ are the working models for modeling the conditional mean of the traits. 
We can re-express the working GLM models (\ref{def:def-beta}) as
\begin{align}
\label{working-model}
y_i= f( \beta_0 + x_i \beta) + \eps_i ~~\text{and}~~ z_i= g(\gamma_0 + x_i \gamma) + v_i,
\end{align}
where $\E[x_i \eps_i] = 0$ and $\E[x_i v_i] = 0$ by taking derivatives in \eqref{def:def-beta}. We mention that including the intercepts $\beta_0$ and $\gam_0$ is crucial when the models are possibly mis-specified as it guarantees $\E[\eps_i]=\E[v_i]=0$. For typical GWAS data sets, the number of genetic variants $p$ is very large, we often assume that the coefficients $(\beta, \gamma)$ in these working models are sparse. 
With linear working models $f(\beta_0 + x_i \beta) = \beta_0 + x_i \beta$ and $g(\gamma_0 + x_i \gamma) = \gamma_0 + x_i \gamma$, the narrow-sense genetic covariance is defined as the bilinear functional $\Cov(x_i \beta, x_i \gamma ) = \beta\tr \Sigma \gamma$. Again, it is defined with respect to the working models instead of the true models.

To make inference of $I^*$ defined in (\ref{def:true-para}), our proposal is based on sample splitting. We use index sets $\cI^0_y$ and $\cI^0_z$ to represent the collected individuals for traits $y$ and $z$, respectively. 
We randomly split samples so that samples for estimating the  coefficients in model \eqref{working-model} and samples for constructing the estimator of covariance are independent. Without loss of generality, we assume that $|\cI^{0}_y|=2n_y$ and $|\cI^{0}_z|=2n_z$. We split $\cI^{0}_y$ into two disjoint subsets $\tilde \I_y$ and $\I_y$ with $|\tilde \I_y| = |\I_y| = n_y$ and split $\cI^{0}_z$ into two disjoint subsets  $\tilde \I_z$ and $\I_z$ with $|\tilde \I_z| = |\I_z| = n_z$.  



Based on the samples in  $\tilde \I_y$ and $\tilde \I_z$, the coefficients ${\beta}$ and ${\gamma}$ are estimated by minimizing the following penalized negative log-likelihood functions 
 \begin{align}
 \label{eq:def-hat-beta-gamma}
 (\hat\beta_0, \hat \beta) & = \argmin_{ \beta_0 \in \cR, \, \beta \in \cR^{p}} |\tilde \cI_y|^{-1} \sum_{i \in \tilde \I_y} \left\{ ( F( \beta_0 + x_i \beta ) - y_i (\beta_0 + x_i \beta) \right\} + \lambda_\beta ( |\beta_0| + \norm{\beta}_1) \, ,\\
  (\hat\gamma_0, \hat \gamma) & = \argmin_{ \gamma_0 \in \cR,\, \gamma \in \cR^{p} } |\tilde \cI_z|^{-1} \sum_{i \in \tilde \I_z} \left\{ G( \gamma_0 + x_i \gamma) -  z_i (\gamma_0 + x_i \gamma) \right\} + \lambda_\gamma (|\gamma_0| + \|\gamma\|_1) , \label{eq:def-hat-beta-gamma2}
 \end{align}
where  $\lambda_{\beta}>0$ and $\lambda_{\gamma}>0$  are tuning parameters. 

\subsection{Estimating the genetic covariance}
\label{sec2-1}

In the second step of our approach,  we  construct the estimator for genetic covariance $I^*$ using samples in $ \I_y$ and $ \I_z$.   First, we plug in the coefficient estimates $\hat \beta$ and $\hat \gamma$, and obtain  the residuals $\hat \varepsilon_i = y_i - \hsfx$ for $i \in \I_y$ and $\hat v_i = z_i - \hsgx $ for $i \in \I_z$. To make use of all observed genotypes, we impute the two traits using the predicted values $\hsfx$ and $\hsgx$ from the working GLMs for all $i \in \I_y\cup \I_z.$ Here $\I_y\cup \I_z$ is the index set of all unique samples and $N = |\I_y\cup \I_z|$ denotes its cardinality. The predicted traits are used to estimate  $\E[\fx \gx]$, which  may differ from the predicted values obtained from the true model. 

In model (\ref{true-model}), $I^*$ can be expressed as (dropping $x_i$ for simplicity)
 $
 I^*=\Cov(f_i^{*}, g_i^{*}) =\E [f_i^{*} g_i^{*}] - \E [f_i^{*}]\E [g_i^{*}].
 $
We  estimate $\E [f_i^{*} g_i^{*}]$ and $\E [f_i^{*}]\E [g_i^{*}]$ separately. We mention that $\E [f_i^{*}]$ and $\E [g_i^{*}]$ are zero in linear models but are possibly nonzero for nonlinear functions $f^*(\cdot)$ and $g^*(\cdot)$, and hence require separate estimates. 
To motivate our estimator of $\E [f_i^{*} g_i^{*}],$ let\\ 
$
  \widehat E^{\text{plug} } =N^{-1} \sum_{i \in \I_y \cup \I_z} { \hfx \hgx }
$
   denote the simple plug-in estimator based on the specified working models. The plug-in estimator $\widehat E^{\text{plug}}$ is sensitive to the  model specification. Specifically, let $f_i$  and $g_i$ denote $\fx$ and $\gx$, respectively. The difference between $\E [f_i^*g_i^*]$ and $\E[f_ig_i]$ is 
\beq
\label{decop}
\E [f_i^*g_i^*] - \E [f_i g_i] = \E[ g_i (f_i^* -  f_i )] + \E[ f_i(g_i^* - g_i )] + \E[(f_i^* - f_i )(g_i^* - g_i )].
\eeq
When $f(\cdot)$ or $g(\cdot)$ is mis-specified, we see from (\ref{decop}) that $\E [f_ig_i]$ is a biased approximation of the  parameter of interest $\E [f_i^*g_i^*]$. To reduce the bias on the right-hand side of (\ref{decop}), we use the empirical residuals $\hat{\eps}_i$ to estimate $f^*_i-f_i$, and $\hat{v}_i$ to estimate $g^*_i-g_i$.

We propose the following estimator of $I^*$,   
\begin{align}
	\label{eq-Ihat}
	& \Iss  = \sum_{i \in \I_y \cup \I_z} \frac{ \hfx \hgx }{N} + \sum_{i \in \I_z} \frac{ \hat v_i \hfx }{n_z} + \sum_{i \in \I_y} \frac{ \hat{\varepsilon}_i  \hgx }{n_y} - \hat{\mu}_{f} \hat{\mu}_g,
\end{align}
where $\hat{\mu}_{f}  = N^{-1}\sum_{i \in \I_y \cup \I_z} { \hfx } + n^{-1}_y \sum_{i \in \I_y} { \hat{\varepsilon}_i }$ and  $\hat{\mu}_{g} = N^{-1}\sum_{i \in \I_y \cup \I_z} { \hgx } + n^{-1}_z\sum_{i \in \I_z} { \hat{v}_i }$ are estimates of $\E[f_i^{*}]$ and $\E[g_i^*]$, respectively. The sum of the first three terms in $\Iss$ is the estimate of $\E[ f_i^{*} g_i^{*}]$. 
The imputation step makes the first term in $\Iss$ an average of $N$ random variables. In general, $N = n_y + n_z - |\cI_y \cap \cI_z|$.  When two traits are collected from the same group of individuals, $\I_y = \I_z$ and $N = n_y = n_z$. When two traits are collected from different groups of individuals, $\I_y \cap \I_z = \emptyset $ and $N = n_y + n_z$. With two traits collected separately, the imputation step is essential because there are no matched pairs, $(f_i, g_i)$ and $(f_i,  v_i)$ or $(g_i, \varepsilon_i),$ available from the data.

\ignore{{\color{magenta}
The proposed estimator $\Iss$ in (\ref{eq-Ihat}) only uses the predicted trait values and residuals. We consider the $\ell_1$-penalized regression in (\ref{eq:def-hat-beta-gamma}) and (\ref{eq:def-hat-beta-gamma2}) for its convenience in prediction and estimation, and its  theoretical guarantees under the sparsity assumptions.  It is however,  possible to apply other penalties or more flexible machine learning methods in the first step. Such machine learning approaches can be theoretically justified if at least one model provides accurate predictions as studied in the double machine learning literature \citep{chernozhukov2018double}.}
}

\subsection{Constructing confidence interval for genetic covariance}
\label{sec2-2}
To construct the confidence interval for $I^*$, we are left to devise a variance estimator for $\widehat I$. However, after considering potential model mis-specification, the analytical expression of variance could be tedious to calculate, if not impossible. We consider an estimate based on the empirical values rather than the limiting distribution. Denote the centered estimated values as 
 \(  f^c(x_i \hat\beta) = \hfx - \hat{\mu}_{f}\) and \( g^c(x_i \hat\gamma) =  \hgx  - \hat{\mu}_{g}.\) Let $\bI(\cdot)$ represent the indicator function and we define the empirical value $\widehat{\Delta}_i$ as 
\beq
\label{eq:delta}
 \widehat{\Delta}_i = \frac{ f^c({x_i} \hat \beta) g^c(x_i \hat\gamma)}{N} + \frac{ \hat{\eps}_i  g^c( x_i \hat\gamma) }{n_y}\bI(i \in \I_y) +\frac{ \hat{v}_i f^c({x_i} \hat \beta)}{n_z} \bI(i \in \I_z). 
 \eeq
 We notice that $\widehat{I}$ can be rewritten as the sum of $\widehat{\Delta}_i$, i.e., $\widehat I = \sum_{i \in \I_y \cup \I_z} \widehat\Delta_i.$  Hence, a natural variance estimator for $\widehat I$ is \(\widehat \sigma^2 =  \sum_{i \in \I_y \cup \I_z} (\widehat{\Delta}_i - N^{-1}{\widehat I})^2\).
Our proposed $(1-\alpha)\times 100\%$ two-sided confidence interval for $I^*$ is defined  as
 \beq
 \label{CI}
 CI(\alpha) = [\widehat I - z_{1 - \alpha/2} \widehat \sigma , \widehat I + z_{1 - \alpha/2} \widehat \sigma ], 
 \eeq
 where $z_{1 - \alpha/2}$ is the upper $\alpha/2$ quantile of standard normal distribution.



We summarize our proposed procedure in Algorithm \ref{alg1}. The method can be extended for estimating the genetic covariance in case-control studies via bias correction of the intercepts and by using weighting to correct the over-sampling of cases (Section A of the Supplemental Materials).  The proposed method can also be used to estimate the heritability by setting $z_i = y_i$. With data from the same group of individuals, the proposed estimator for heritability generalizes \cite{cai2018semi} to nonlinear trait models. 

 \medskip
\begin{algorithm}[H]
 \SetKwInOut{Input}{Input}
\SetKwInOut{Output}{Output}
\SetAlgoLined
 \Input{$(y_i, x_i)$ for $i \in \cI^0_y$, $(z_i, x_i)$ for $i \in \cI^0_z$, $\lambda_\beta$,  $\lambda_\gamma$, confidence level $\alpha$. }
 \Output{$\widehat I$ and $CI(\alpha)$.}
\caption{Estimation and inference for genetic covariance}

1. Split $\cI^0_y$ into $\tilde \I_y$ and $\I_y$. Split $\cI^0_z$ into $ \tilde \I_z$ and $\I_z$. 

2.  Estimate the $\hat \beta$ and $\hat \gamma$ as in \eqref{eq:def-hat-beta-gamma} based on samples in $\tilde \I_y$ and $\tilde \I_z$.

3. Calculate $\hat \varepsilon_i = y_i - \hsfx$ for $i \in \I_y$ , $\hat v_i = z_i - \hsgx $ for $i \in \I_z$. Estimate the genetic covariance $\Iss$ as in \eqref{eq-Ihat}.

4. Calculate $\hat \Delta_i$ via \eqref{eq:delta} and empirical variance estimator $\widehat \sigma^2$. Construct $CI(\alpha)$ as in \eqref{CI}.
\label{alg1}
\end{algorithm}

\subsection{Inference for the narrow-sense genetic covariance under linear working models}
\label{sec2-5}
If the narrow-sense genetic covariance $\beta^{\intercal}\Sig\gamma$ is used,  its estimation and inference can be achieved by choosing the working models $\hfx = \hat \beta_0 + x_i \hat \beta $ and $ \hgx = \hat \gamma_0 + x_i \hat \gamma$. In this case, $\Iss$ can be viewed as a de-biased estimator of $\beta \tr \Sigma \gamma$. Since the linear working models may be mis-specified, the proposed empirical variance estimator $\widehat{\sig}^2$ is crucial. This is different from the setting of correctly specified linear models, where inference methods can be derived based on the limiting distribution \citep{janson2017eigenprism,cai2018semi}.

\section{Asymptotic normality and robustness}
\label{sec:theory}

 Let $s_\beta = |\supp(\beta)| + \bI(\beta_0  \neq 0 ) $ and  $s_\gamma = |\supp({\gamma})| + \bI(\gamma_0  \neq 0)$ denote the sparsity of coefficients defined in \eqref{working-model}. 
We state following assumptions for theoretical analysis.
\begin{Assumption}
\label{asm: assumption}
 \label{C1:var-cond} The observed variants $x_i$ are $i.i.d.$ samples with mean zero, $\E(x_i \tr x_i) = \Sigma$, and bounded sub-Gaussian norm. The eigenvalues of $\Sig$ satisfy $ C \leq  \lambda_{\min} (\Sigma) \leq  \lambda_{\max}(\Sigma) \leq C^{-1}<\infty.$
\end{Assumption}

\begin{Assumption}
\label{asm: assumption-y}
 \label{C2:var-cond} The traits $y_i$ and $z_i$ are sub-Gaussian random variables with the bounded sub-Gaussian norm. The error terms $\eps^*_i$ and $v^*_i$ are  sub-Gaussian random variables with mean zero and the bounded sub-Gaussian norm. Besides, $ \min \{\Var(\eps^*_i) , \Var(v^*_i) \} \geq c > 0 $. 
\end{Assumption}

  \begin{Assumption}
\label{C2:GLM-condition} The functions $F(t)$ and $G(t)$ are twice differentiable with $f(t) = F^{\prime}(t)$ and $g(t) = G^\prime(t)$.  The derivatives of $f(t)$ and $g(t)$ are uniformly bounded, i.e., $|f^\prime(\cdot)| \leq C<\infty $ and $ |g^\prime(\cdot)| \leq C<\infty.$ Moreover, 
 $f^\prime(t)$ and $g^\prime(t)$ satisfy the Lipschitz condition for a positive constant $L > 0,$
 \[
\max_{i \in \cI^0_y \cup \cI^0_z} |f^\prime(\beta_0 + x_i \beta) -  f^\prime(\beta_0 + x_i \beta + \theta)| \leq  L |\theta|, ~~ \max_{i \in \cI^0_y \cup \cI^0_z} |g^\prime( \gamma_0 + x_i \gamma ) -  g^\prime(\gamma_0 + x_i \gamma + \theta)| \leq  L |\theta|.
  \]
\end{Assumption}


 Our  method allows the collected individuals to have missing outcomes. 
	Let $M_i$ be the indicator of missingness. 
	Assumption 1 and Model (1) together imply that the distribution of genotypes does not depend on the missing status, 
$P(x_i | M_i = 1) = P(x_i | M_i = 0)$ and the conditional mean function of the outcome is the same regardless of the missing status, $\E[y_i | x_i, M_i = 1] = \E[y_i | x_i, M_i = 0] $ and $\E[z_i | x_i, M_i = 1] = \E[z_i | x_i, M_i = 0].$ Both conditions are readily satisfied in most genetic studies. 
The boundedness assumption on $f'(t)$ and  $g'(t)$ has  been used in  \citet{negahban2012unified} and \citet{loh2015regularized}. The Lipschitz continuity for $f'(t)$ and $g'(t)$ is also  required in \cite{van2014asymptotically} for inference in the high-dimensional GLMs. 
Simple calculations show that the standard linear model, logistic model, and multinomial logistic model satisfy Assumption \ref{C2:GLM-condition}. Besides, in the following,  we require  $\min\left\{ \Var( f^*(x_i) ) ,  \Var( g^*(x_i) ) \right\} \geq  c > 0$ to ensure the variations of two traits are sufficiently explained by the variants $x_i$.  This condition is the foundation of further discussion on genetic covariance. 

\subsection{Theoretical properties under correct model specification}
\label{sec:correct}
Let $I_n$ denote a counterpart of $\Iss$ based on the population parameters $\beta$ and $\gam$. That is,
 \begin{equation}
 \label{I-tilde}
  I_n =  \sum_{i \in \I_y \cup \I_z} \frac{ f^c(\beta_0 + x_i \beta) g^c(\gamma_0 + x_i \gamma) }{N} + \sum_{i \in \I_z} \frac{ v_i f^c(\beta_0 + x_i \beta)  }{n_z} + \sum_{i \in \I_y} \frac{  \varepsilon_i g^c(\gamma_0 + x_i \gamma) }{n_y},
 \end{equation}
where $f^c(\beta_0 + x_i \beta) = \fx - \mu_f$ and  $g^c(\gamma_0 +x_i \gamma) = \gx - \mu_g$. In fact, the variance of $ I_n$ is same as the variance of $\Iss$.  Let $\nmax = \max\{n_y, n_z\}$.

In Theorem \ref{thm1}, we first present  the convergence rate of $\Iss$ when both models are correctly specified. Tuning parameters are chosen as $\lambda_\beta = c_1 \sqrt{\log p/n_y} $ and $ \lambda_\gamma = c_2 \sqrt{\log p/n_z}$, where $c_1$ and $c_2$ are large enough positive constants.
\begin{Theorem}
\label{thm1}
\label{thm:asym-normality}
Suppose both models are correctly specified, i.e., $f^*(x_i) = \fx $ and $g^*(x_i) = \gx$. 
Assume Assumptions \ref{C1:var-cond} - \ref{C2:GLM-condition} and $\max \{s_\beta \log p/n_y  , s_\gamma \log p/n_z \} = o(1)$. Then  with probability at least $1 - p^{-c_0} - e^{-cn} -  Ce^{-t^2}$,
\begin{align}
\label{eq:convergence-rate-general}
|\Iss - I^*| & \lesssim  \left(  n^{-1/2}_z + n^{-1/2}_y  \right)t + R_n \quad \text{where} \quad R_n = \sqrt{s_{\beta} s_{\gamma}} \frac{\log p}{\sqrt{n_z n_y}} + o(n_y^{-1/2} + n_z^{-1/2} ) .
\end{align}
If $ \sqrt{s_\beta s_\gamma}  \log p = o\left({\sqrt{\nmax}} \right) $, 
then we have \( (\Iss - I^*)/\sqrt{ \Var(  I_n) } \cid N(0, 1). \)
 
\end{Theorem}

Theorem \ref{thm1} implies that  under the mild condition  $\max \left\{s_\beta \log p/n_y, s_\gamma \log p/n_z \right\} = o(1)$,  the consistency of $\Iss$ holds. The estimation bias $R_n$ is determined by the estimation accuracy of $\hat \beta$ and $\hat \gamma$. For asymptotic normality, the sparsity condition $ \sqrt{s_\beta s_\gamma}\log p = o\left({\sqrt{\nmax}}\right) $ is needed.  
If $s_{\beta}\asymp s_{\gam}$ and $n_y\asymp n_z$, then this condition reduces to $s_{\beta}\log p=o(\sqrt{n_y})$ and $s_{\gam}\log p=o(\sqrt{n_z})$, which are the so-called ultra-sparse conditions. 
If $s_{\beta}\gg s_{\gamma}$ and $n_y\gg n_z$, then a sufficient condition is $s_{\beta}\log p=O(n_y/\sqrt{n_z})$ and $s_{\gam}\log p=o(\sqrt{n_z})$. That is, if the relatively small sparsity $s_{\gam}$ is in the ultra-sparse regime, then the relatively large sparsity $s_{\beta}$ can be beyond the ultra-sparse regime. The asymptotic variance $\Var(I_n)$ is of order $1/\min\{n_y,n_z\}$ as implied by (\ref{eq:convergence-rate-general}). Under the same conditions, we show the consistency of the variance estimator $\widehat{\sig}^2$  and the asymptotic validity of the confidence interval in Theorem S1 in the Supplemental Materials.

To illustrate the benefits of sample splitting, we provide a theoretical analysis of the full-sample estimator $\widehat{I}_{\text{full}}$,  which is computed with Algorithm \ref{alg1} with $ (\cI_y, \cI_z) = (\tilde \cI_y, \tilde \cI_z)=(\cI^0_y,\cI^0_z)$.
\begin{Theorem}
	\label{thm1-full}
	Suppose $n_y \asymp n_z \asymp n$ and $(\fx, \gx, v_i, \eps_i )$ have bounded values. When both models are correctly specified, i.e., $f^*(x_i) = \fx $ and $g^*(x_i) = \gx$, then under Assumptions \ref{C1:var-cond} - \ref{C2:GLM-condition} and condition $\max \{s_\beta \log p/n_y  , s_\gamma \log p/n_z \} = o(1)$, for the full-sample estimator, it holds  with probability at least $1 - p^{-c_0} - e^{-cn} -  Ce^{-t^2}$ that,
	\begin{align}
		\label{eq:convergence-rate-full}
		|\widehat{I}_{\textup{full}} - I^*| & \lesssim {  \frac{1}{\sqrt{n}} t } + R_n   ~\text{ where }~ R_n =   \frac{ \max(s_\beta, s_\gamma)  \log p }{n}.
	\end{align}
If $\max \{s_\beta, s_\gamma \} {\log p} = o\left({\sqrt{n}} \right)$, then $ R_n/\sqrt{ \Var(  I_n) } = o(1) $ and $(\widehat{I}_{\textup{full}} - I^*)/ \sqrt{ \Var( I_n) } \cid N(0, 1).  $
\end{Theorem}
Theorem \ref{thm1-full} considers the setting where the sample sizes for traits $y$ and $z$ are asymptotically balanced.
In view of \eqref{eq:convergence-rate-full}, the estimation consistency of $\widehat{I}_{\textup{full}}$ is guaranteed when  $\max\{s_{\beta}, s_{\gamma} \} \log p =o(n)$.  However, a stronger condition $\max \{s_\beta, s_\gamma \} \log p = o\left(\sqrt{n}\right)$ is needed for asymptotic normality. Comparing with the sparsity conditions in Theorem \ref{thm1} with $n_y\asymp n_z \asymp n$, two sets of conditions are equivalent if $s_{\beta}\asymp s_{\gam}$. In the unbalanced sparsity regime, $\Iss$ has smaller bias and requires milder conditions for its asymptotic normality. Specifically, if $ s_{\beta} \gg s_{\gamma} \asymp 1$, then $\widehat{I}_{\textup{full}}$ requires $ s_\beta \log p =o(\sqrt{n} ) $ and $\Iss$ requires  $s_\beta (\log p)^2 =o(n)$ for asymptotic normality.  The improvement of the split-sample estimator  in the unbalanced sparsity settings is also supported by the simulation results.

%
We conclude  this section by emphasizing the importance of imputing the trait values for samples with the missing outcomes. In our procedure, we use $\hfx$ as an approximate for the unobserved $y_i$ for $i\in \I_z\setminus \I_y$, and using $\hgx$ as an approximate for  the unobserved $z_i$ for $i\in \I_y\setminus \I_z$. 
In the second step of our estimation,  if we  only uses overlapping samples $\cI_o = \cI_y \cap \cI_z$ without imputing the missing traits,  
when  $n_y \asymp n_z \asymp n$ and $n_o = |\cI_o|$, the estimation error  is $1/ \sqrt{n_o}+ \sqrt{s_\beta s_\gamma  }\log p/ n $ with sample splitting and  is  $1/ \sqrt{n_o}+ \max(s_\beta, s_\gamma )\log p/ n $
without sample splitting.  We see that  using only the overlapping samples with both traits observed  leads to a  larger variance when $n \gg n_o.$ More detailed discussion can be found in Supplemental Materials.  

\subsection{Robustness to model mis-specification}
\label{sec:robust}
In practice, the estimates may be biased if the assumed working models are poor approximations of $f^*(x)$ and $g^*(x)$. 
This subsection shows the robustness of the proposed method when the working models are possibly mis-specified. 
Recall that $I_n$ in (\ref{I-tilde}) is the counterpart of $\Iss$ based on the population parameters $\beta$ and $\gam$, and the expectation of $I_n$ is
\begin{equation}
 \label{working-linear-model}
 \E [I_n]= \E\left[ f^c(\beta_0 + x_i \beta) g^c(\gamma_0 +x_i \gamma) + \eps_i g^c(\gamma_0 +x_i \gamma) + v_i f^c(\beta_0 + x_i \beta)  \right].
 \end{equation}
We note $\E [I_n]$ is different from the true parameter $I^*$ in general when $f(\cdot)$ or $g(\cdot)$ is mis-specified. 
The following proposition reveals the doubly robustness of $\E [I_n]$ for approximating $I^*$.
 
\begin{Proposition}
\label{prop:I-equiv-1}
Under Assumption \ref{C1:var-cond}, for $\E [I_n]$ defined in (\ref{working-linear-model}) and true parameter $I^*$ defined in (\ref{def:true-para}),  it holds that 
\(
 I^* - \E [I_n] = \E\left[\left( \gx - g^{*}(x_i)\right) \left( \fx - f^{*}(x_i)\right)\right]. \)
If at least one model is correctly specified, i.e., $f^*(x_i) = \fx $ or $g^*(x_i) = \gx,$ then 
\(
\E [I_n]  = I^*.
\)
\end{Proposition}

Proposition \ref{prop:I-equiv-1} shows that the true parameter $I^*$ and $\E[I_n]$ are equivalent if at least one conditional mean model is correctly specified. This property comes from the fact that we leverage the residuals to ``de-bias'' in $\E[I_n]$. That is, we use the last two terms on the right-hand side of (\ref{working-linear-model}) to correct the potential bias of the first term. This doubly robustness property preludes the robustness of $\Iss$ to model mis-specification under certain conditions, as shown in Theorem \ref{thm2}.
\begin{Theorem}
\label{thm2}
Assume that $f^*(x_i) = f(\beta_0 + x_i \beta)$ and $g^*(x_i) \neq g( \gamma_0 + x_i \gamma)$ . 
If Assumptions \ref{C1:var-cond} - \ref{C2:GLM-condition} and $\max \{s_\beta \log p/n_y  , s_\gamma \log p/n_z \} = o(1)$ hold, then
\(
|\Iss - I^{*}|  = o_p(1).
\)
\end{Theorem}
 The robustness property proved in Theorem \ref{thm2} also applies to the full-sample estimator $\widehat{I}_{\textup{full}}$ under the same conditions. The doubly robustness property also allows us to consider the double machine learning framework to replace the penalized GLMs when sparsity assumptions can be dubious as we will discuss in Section \ref{sec:discuss}. 
Although the population-level equivalence of $\E[I_n]$ and $I^*$ holds with one correctly specified model,
it does not imply that asymptotic normality of $\Iss$ directly hold in this case. This is because model mis-specification can invalidate some moment equations and the remaining bias, $\Iss - I_n$, can be larger under model mis-specification. For example, $\E[v_i|x_i]=0$ when $g^*(\cdot)=g(\cdot)$ but we only have $\E[v_ix_i]=0$ when $g^*(\cdot)\neq g(\cdot)$.

We next consider a special case of mis-specification, where one linear model is correctly specified, $f^*(x_i) = f(\beta_0 + x_i \beta) = \beta_0 + x_i \beta$,  and $g(\gam_0+x_i\gam)$ can be nonlinear and mis-specified. 
In this case, the linear form of the correctly specified $f^*(x_i)$ guarantees the moment condition $\E[v_i f^*(x_i)] = \E[v_i ( \beta_0 + x_i \beta)] = 0$, which allows us to derive the asymptotic normality.

\begin{Theorem}
\label{thm:I-equiv-2}
Assume $f^*(x_i) = f(\beta_0 + x_i \beta) = \beta_0 + x_i \beta$ and $g^*(x_i) \neq g( \gamma_0 + x_i \gamma)$.  Suppose that Assumptions \ref{C1:var-cond} - \ref{C2:GLM-condition}, $\max \{s_\beta \log p/n_y  , s_\gamma \log p/n_z \} = o(1)$, and $\min\left\{ \Var( f(x_i \beta) ) ,  \Var( g(x_i \gamma)) \right\} \geq c > 0$ hold. For the proposed split-sample estimator, if  
\begin{equation}
\label{eq:cond2}
  s_\beta {\log p} = o\left({ \nmax^{1/2} } \right)~~\text{and}~~  s_{\gamma} {\log p} = O\left( {n^{2/3}_{z}  } \right),
 \end{equation} 
 then \( {(\Iss -  I^*) }/{ \sqrt{ \Var(  I_n) } } \cid N(0, 1) \) and \( \lim_{n_y, n_z \rightarrow \infty} \mathbb{P}( I^* \in CI(\alpha)) = 1 - \alpha .\)
\end{Theorem}
Theorem \ref{thm:I-equiv-2} shows that the genetic covariance can be robustly inferred with at least one true linear trait. We prove the asymptotic normality of $\Iss$ under conditions where $s_{\gam}$ could be large due to the model mis-specification. The condition in \eqref{eq:cond2} assumes mild sparsity for the possibly mis-specified GLM, $s_{\gam}$, and a relatively stronger sparsity assumption for the correctly specified linear model, $s_{\beta}$. The condition on $s_{\beta}$ is no stronger than the ultra-sparse condition while the condition on $s_{\gam}$ is much weaker for statistical inference.  We comment that the condition in Theorem \ref{thm:asym-normality}  is also a sufficient condition, but the condition in \eqref{eq:cond2} is weaker with large $s_{\gam}$. 

\subsection{Narrow-sense genetic covariance}
\label{sec:narr-gecv}

The proposed method with linear working models $\fx = \beta_0 + x_i \beta$ and  $\gx = \gamma_0 + x_i \gamma$ can also be used to estimate and make inferences for the narrow-sense genetic covariance $\beta^{\intercal}\Sig\gam$, where $\beta$ and $\gam$ are defined via (\ref{def:def-beta}). Recall that $\E[I_n]$ in \eqref{working-linear-model} is the probabilistic limit of $\Iss$ when $\hat{\beta}$ and $\hat{\gam}$ converge. The following proposition shows the connection between $\Iss$ and $\beta^{\intercal}\Sig\gam$.

\begin{Proposition}
\label{prop:narrow-sense}
$\E[I_n]=\beta^{\intercal}\Sig\gam$ under Assumption \ref{C1:var-cond}.
\end{Proposition}

Proposition \ref{prop:narrow-sense} implies that $\Iss$ converges to $\beta^{\intercal}\Sig\gam$ as long as $\hat{\beta}$ and $\hat{\gam}$ converge, no matter the linear working models are correctly specified or not. This result is a direct consequence of $\E[v_i x_i] = 0$ and $\E[\eps_i x_i] = 0$. 
Next, Theorem \ref{thm:narrow-sense} shows that $\Iss$ provides asymptotically valid inference on the narrow-sense genetic covariance $\beta\tr\Sig\gam$, no matter the linear models are correctly specified or not. 

%

%

\begin{Theorem}
	\label{thm:narrow-sense}
Consider $\Iss$ with linear working models. Suppose that Assumptions \ref{C1:var-cond} - \ref{C2:GLM-condition},  $\min\left\{ \beta\tr \Sigma \beta,  \gamma\tr \Sigma \gamma \right\} \geq c > 0$ and $\max\{s_\beta \log p/n_y, s_\gamma \log p/n_z \} = o(1)$ are satisfied. If 
	\begin{equation}
  \label{cond-ns}
	\text{(i) }	s_{\beta}  = o\left(\frac{{ n^{1/2}_{\max} } }{\log p} \right) ~ \text{and} ~ s_{\gam} = O\left( \frac{ n^{3/4}_z }{\log p} \right) \text{, or (ii) } s_{\gamma}  = o\left(\frac{{ n^{1/2}_{\max} } }{\log p} \right) ~ \text{and} ~ s_{\beta} = O\left( \frac{ n^{3/4}_y }{\log p} \right),
	\end{equation}
	then $(\Iss - \beta^{\intercal}\Sig\gam )/\sqrt{ \Var(  I_n) }  \cid N(0, 1) $ and $\lim_{n_y, n_z \rightarrow \infty} \mathbb{P}\left( \beta^{\intercal}\Sig\gam \in CI(\alpha)\right) = 1 - \alpha.$
\end{Theorem} 

Thanks to the simple linear form, for the larger sparsity, say, $s_\beta$,  we get a weak condition $s_{\beta} {\log p} = O({ n^{3/4}_{y} })$ for asymptotic normality. Comparing with the inference results for the general quadratic functional $I^*$, the inference for the bilinear form $\beta\tr\Sig\gam$ requires weaker conditions.



As a byproduct, Proposition \ref{prop:I-equiv-1} and Proposition \ref{prop:narrow-sense} together  imply  that the narrow-sense genetic covariance $\beta^{\intercal}\Sig\gam$ reflects the total genetic covariance when at least one trait is linear. 

\begin{Corollary}
\label{coroll:narrow-sense}
Under Assumption \ref{C1:var-cond}, for $\beta$ and $\gam$ defined via (\ref{def:def-beta}) with linear $f(\cdot)$ and $g(\cdot)$,  $I^* = \beta^{\intercal}\Sig\gam$ when at least one model is correctly specified. 
\end{Corollary}

\section{GWAS simulations and evaluation of the methods}
\label{sec:simulation}
In this section, we evaluate the numerical performance of the proposed methods. To simulate data that mimics real GWAS data, we use phenotypic values generated from the real genotypic data  of a GWAS of pediatric autoimmune diseases \citep{YunLi}. This dataset includes $10718$ subjects in the control group with a total of $475324$ SNPs genotyped on 22 autosomes. The potential causal loci is selected from $K = 13649$ genetic variants after LD-based pruning in \texttt{plink} software. In each experiment, we repeat the simulations $200$ times and in each repetition, we randomly select 8000 individuals to generate the traits, including the overlapping case where  8000  individuals have both two traits measured and the non-overlapping case where  4000 individuals have one of the two traits measured respectively.  

We use the \texttt{bigstatsr} package \citep{prive2018efficient} to fit the penalized regression for the large dataset. Tuning parameters are selected based on the 10-fold cross-validation. To evaluate the performance of the full-sample and the split-sample estimators, we split $90\%$ samples for  estimating the coefficients $(\hat \beta, \hat \gamma).$ Then the full-sample estimator is constructed from the same data  and the split-sample estimator is constructed from the remaining $10 \%$ samples.
The estimates from the random-effects variance component model using  GCTA \citep{lee2012estimation} are compared.

\subsection{Simulations with continuous traits}
\label{sec:LL}
We consider  continuous traits  generated from the following  models  
\begin{equation}
\label{eq: simulate double models}
y_i =  f^{*}(x_i) + \varepsilon^*_i  ~~ \text{and} ~~  z_i =  g^{*}(x_i)  + v^*_i. 
\end{equation} 
Let $\Var(y_i) = \Var(z_i) = 1$ and the heritability  $\Var(f^{*}(x_i)) =  \Var(g^{*}(x_i)) = 0.5.$ The error terms \(\varepsilon^{*}_i\) and \(v^{*}_i\)  are assumed to follow a bivariate normal distribution with $\E(\varepsilon^{*}_i) = \E(v^{*}_i) = 0$, $\Var(\varepsilon^{*}_i) = \Var(v^{*}_i) = 0.5$ and $\Cov(v^{*}_i, \varepsilon^{*}_i) = 0.2.$ 
To generate the traits with different genetic architectures, we randomly select causal SNP sets $C_\beta$ and $C_\gamma$ with size $|C_\beta| = \prop_\beta \times K$ and $|C_\gamma| =  \prop_\gamma \times K.$ The overlapping genetic architecture is determined by the overlapping set of  variants  $ C_o  =  C_\beta \cap  C_\gamma$ with its size of $\prop_o \times K.$ 
By specifying different proportions $(\prop_\beta, \prop_\gamma, \prop_o )$, various genetic architectures are considered, including 
both traits generated from sparse models or from polygenic models, and 
one trait generated from a sparse model while the other one generate from a polygenic model. We also consider that model when  both traits  are generated from polygenic models that include both dense weak effects and sparse strong effects, and  the genetic covariance is contributed by sparse and strong effects.

We consider the following  four  specific trait  models  $f^*(x_i)$ and $g^*(x_i)$: 

\noindent {\bf (a)}. True linear models $f^{*}(x_i) = \sum_{i \in C_\beta} x_i \beta_i$ and $g^{*}(x_i) = \sum_{i \in C_\gamma} x_i \gamma_i$ with non-zero coefficients generated from the normal distributions and  the coefficients are then re-scaled to guarantee $I^* = 0.2.$ Details can be found in  Section B of the Supplemental Materials.

\noindent {\bf (b)}. True linear models with the genetic effects distribution  dependent on the minor allele frequencies and its linkage disequilibrium (LD) measure based on the LDAK model \citep{speed2012improved}.  Specifically, for the $i$th variant,  we define the weight $h_i = (f_i - f^2_i)^{-0.75}/\ell_i,$ where $f_i$ is the effect  allele frequency and $\ell_i$ is the calculated LD score. The genetic effects are then specified as  $\tilde \beta_i = h_i \beta_i$ and $\tilde \gamma_i = h_i \gamma_i$ with $( \beta_i, \gamma_i)$ generated from normal coefficients. It implies that, at a given minor allele frequency (MAF), low-LD SNPs have larger effect sizes and at a given LD, SNPs with lower MAF have larger effect sizes.

\noindent  {\bf (c)}. Same model as (b) while $h_i = (f_i - f^2_i)^{0.75} \times \ell_i.$ It implies that, at a given minor allele frequency, high-LD SNPs have larger effect sizes and  at a given LD, SNPs with higher MAF have larger effect sizes.

\noindent {\bf (d)}.  The trait  $z_i$ is generated from the true linear model as  in (a) while the continuous trait $y_i$ is generated from the composite model  $ f^{*}(x_i) = \sum_{j \in C_\beta} ({x^{1/2}_{ij}} + x^3_{ij} ) \beta_j +  \sum_{ j \in C_\beta } x_{ij} x_{i j'} \beta_j \beta_{j'},$
    where $j'$ is the randomly selected index of the variant that  interacts with  the $j$th variant.

\ignore{
\item[(e)]  trait $z_i$ is generated from the true linear model while the continuous trait $y_i$ is generated from the  model  with dominant effects $$ f^*(x_i) = \sum_{j \in C_\beta } \mathbbm{1}( x_{ij} > 0) \beta_j.$$

  \item[(f)] Both models are generated from the dominant effects model defined in as in model (e). 
}

\begin{figure}
\resizebox{0.95\linewidth}{!}{
\begin{tabular}{cc}
  \includegraphics[height=0.3\textheight, width=0.5\textwidth]{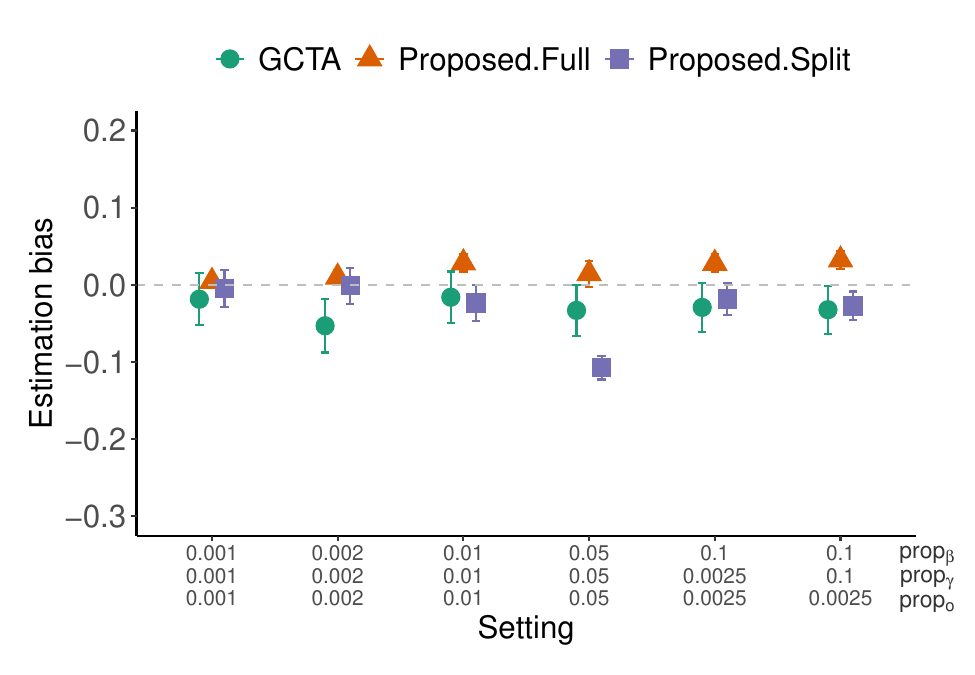} &  
  \includegraphics[height=0.30\textheight,width=0.5\textwidth]{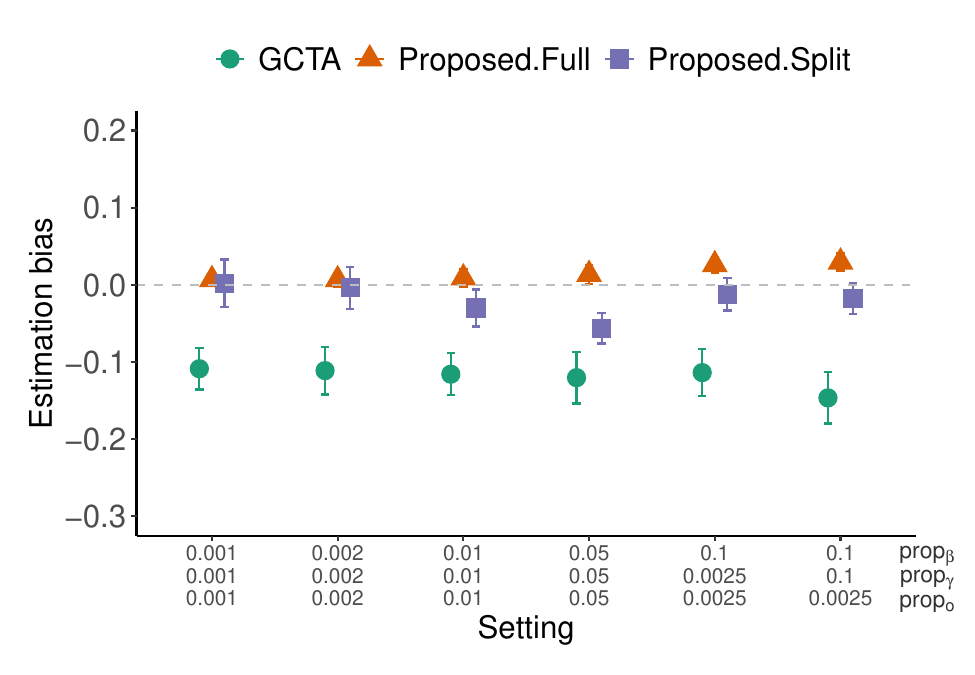} 
  \vspace{-0.1in}\\
  \vspace{-0.02in}
  (a) Normal coefficients, linear models & (b) LDAK 1, linear models \\
  \includegraphics[height=0.30\textheight, width=0.5\textwidth]{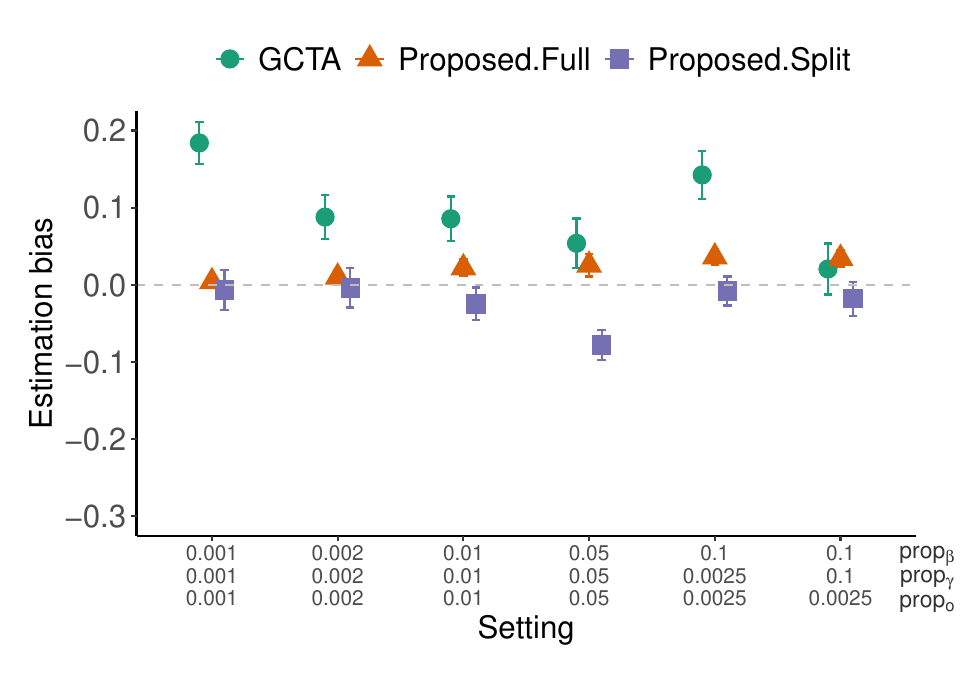} &  
  \includegraphics[height=0.30\textheight,width=0.5\textwidth]{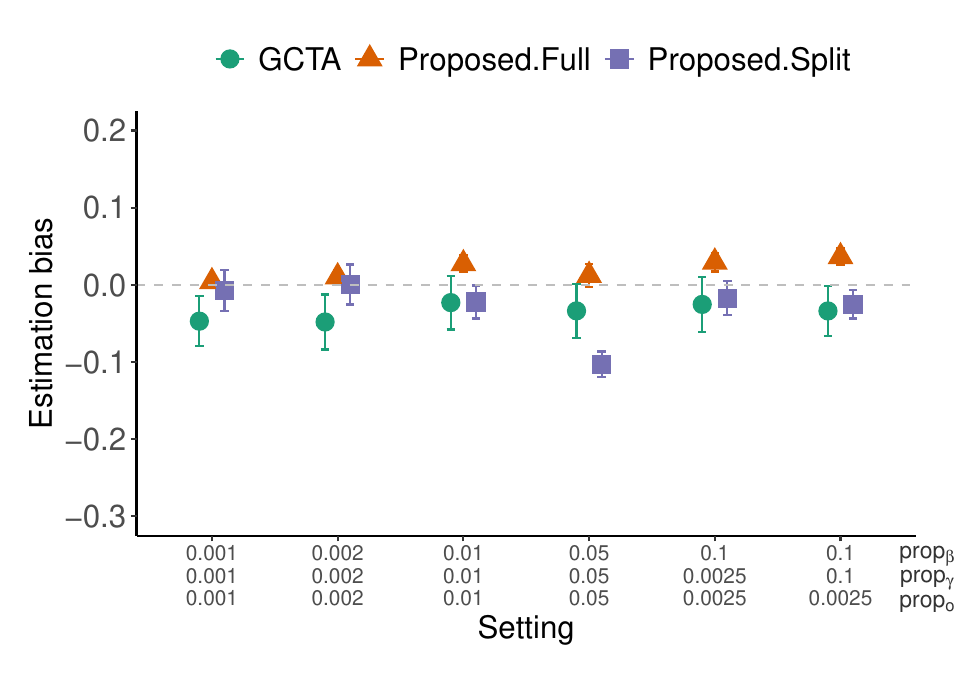} 
  \vspace{-0.1in}\\
  \vspace{-0.02in}
  (c) LDAK 2, linear models & (d) Single mis-specified composite model \\
  \ignore{
  \includegraphics[height=0.30\textheight,width=0.5\textwidth]{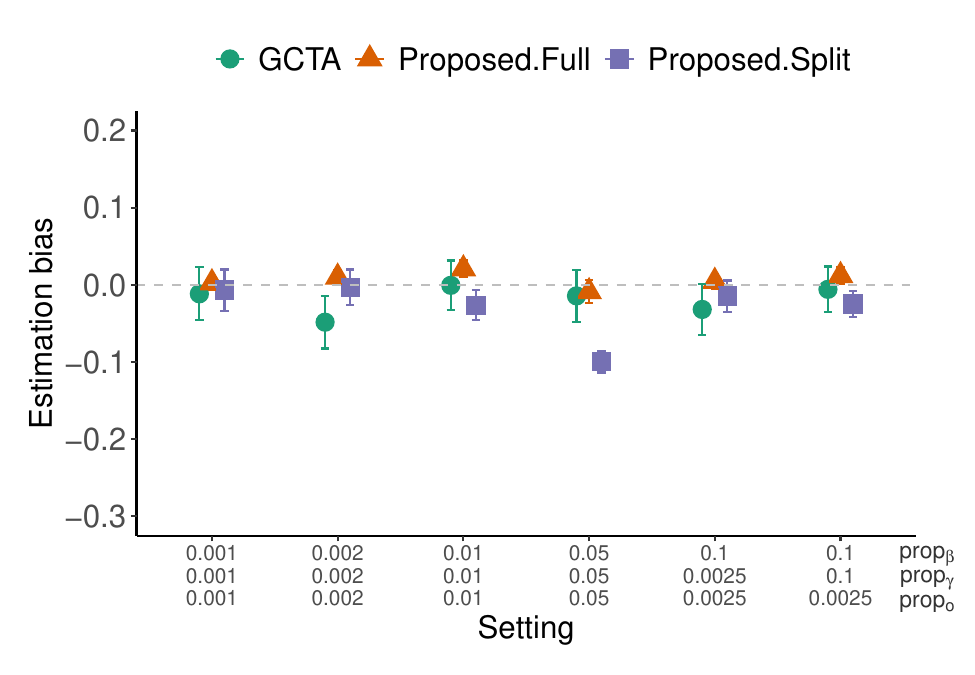} &  
  \includegraphics[height=0.30\textheight,width=0.5\textwidth]{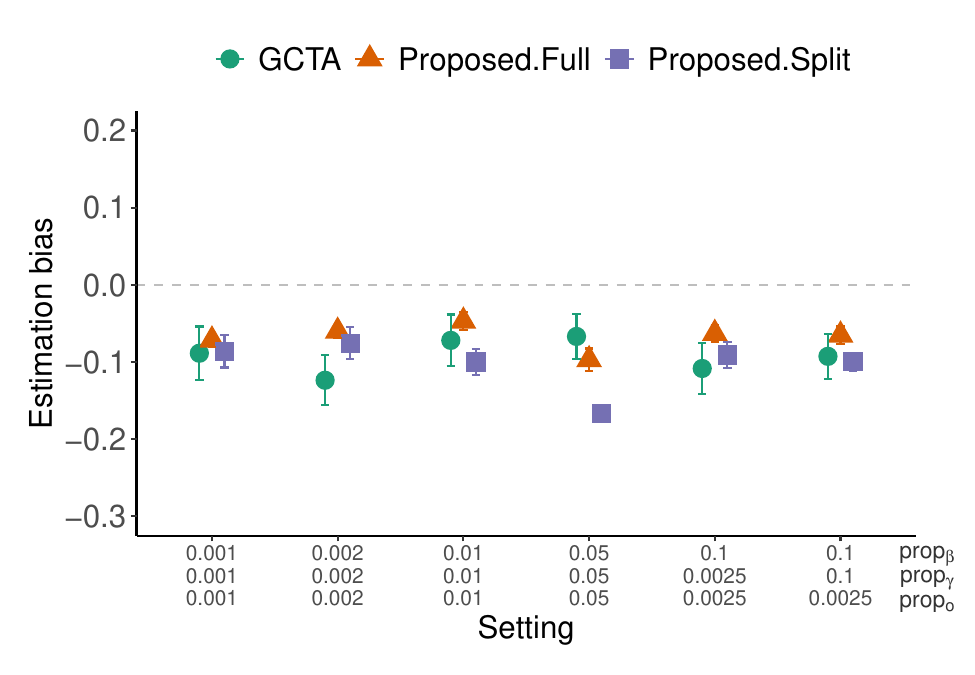} \\
  (e) Single mis-specified dominant model  & (f) Two mis-specified dominant models \\}
\end{tabular}
}
\caption{Comparison of estimation accuracy for four different genetic models  when both traits are measured on the same set of individuals ($N_y = N_z = 8000$).  Each point represents the mean of the average bias, with error bars representing the standard error on both sides. 
}
\label{fig:LL-over}
\end{figure}

\ignore{
\begin{figure}
\resizebox{\linewidth}{!}{
\begin{tabular}{cc}
  \includegraphics[height=0.3\textheight, width=0.5\textwidth]{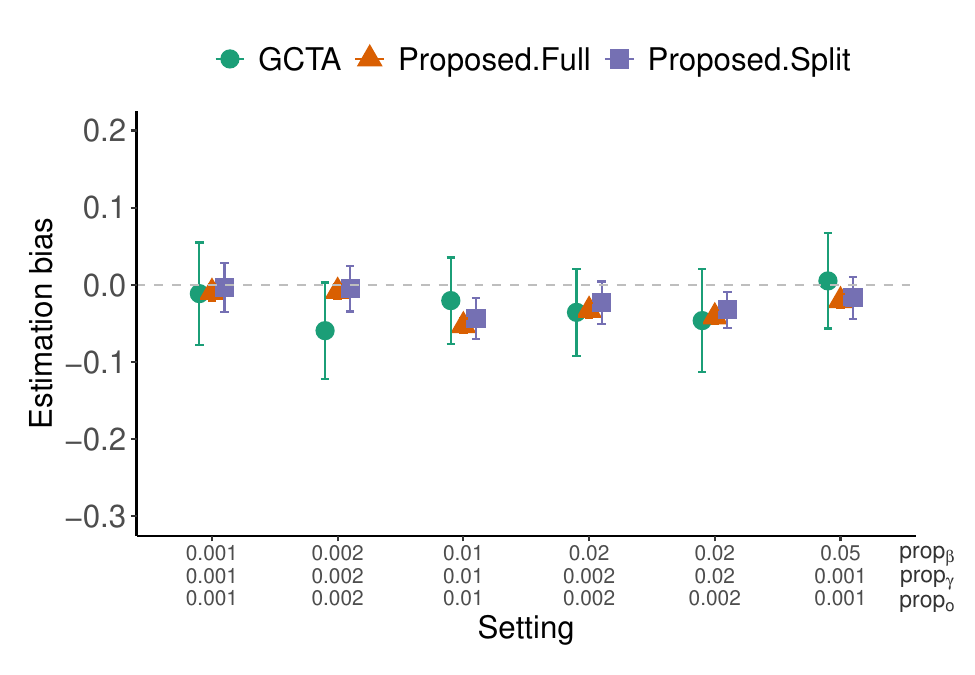} &  
  \includegraphics[height=0.30\textheight,width=0.5\textwidth]{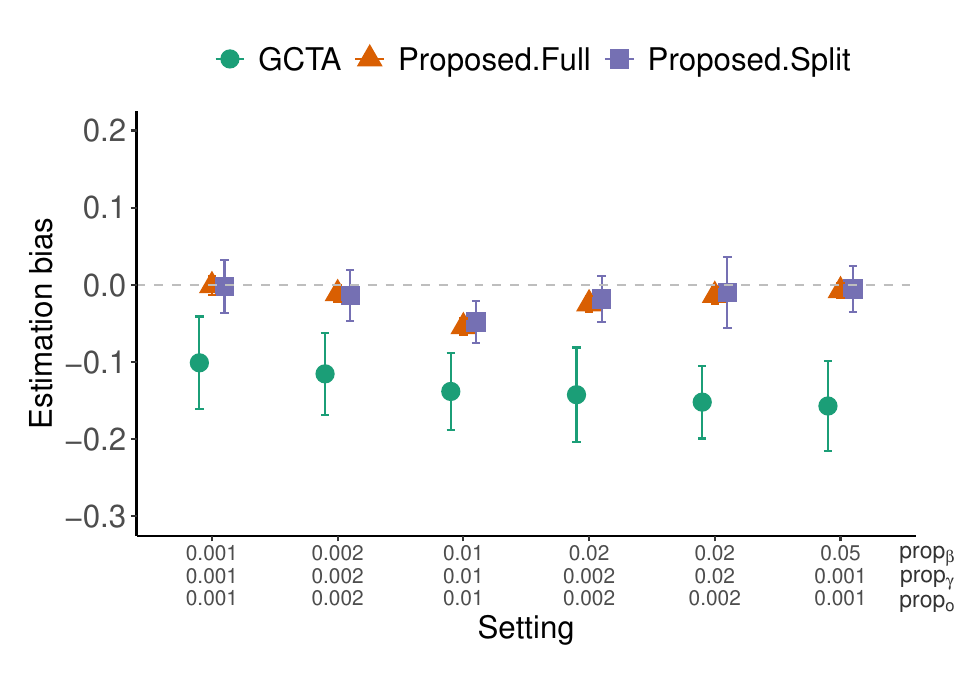} \\
  (a) Normal coefficients, linear models & (b) LDAK 1, linear models \\
  \includegraphics[height=0.30\textheight, width=0.5\textwidth]{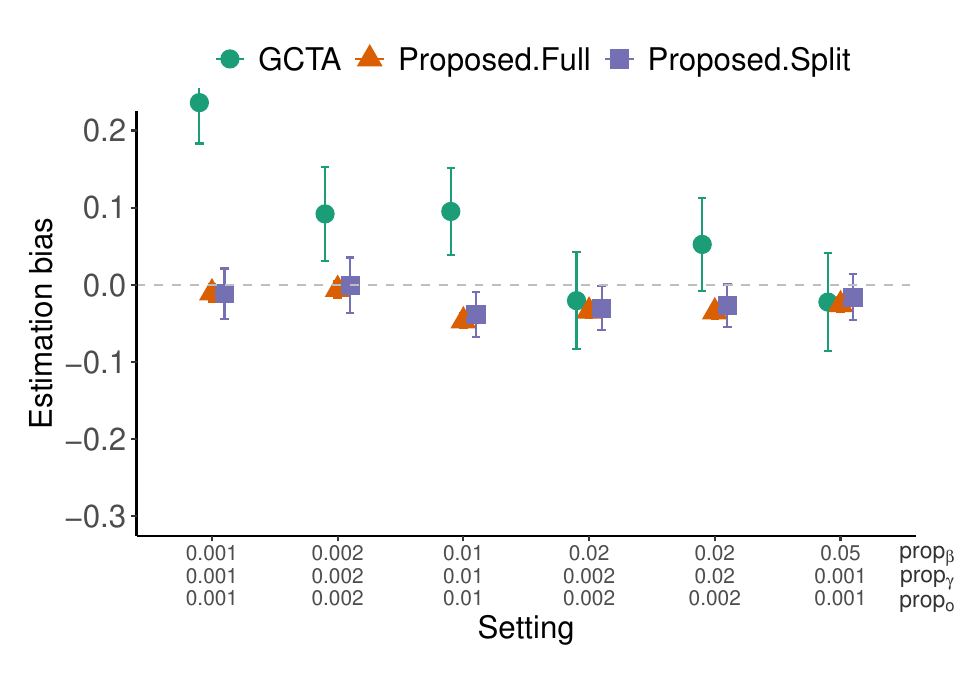} &  
  \includegraphics[height=0.30\textheight,width=0.5\textwidth]{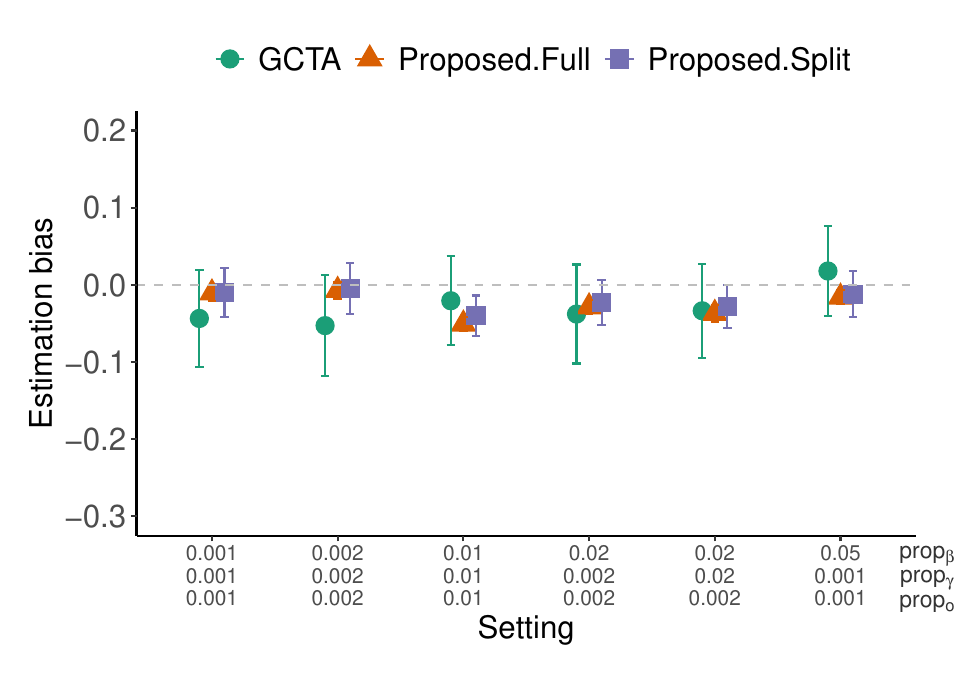} \\
  (c) LDAK 2, linear models & (d) Single mis-specified composite model \\
  \includegraphics[height=0.30\textheight,width=0.5\textwidth]{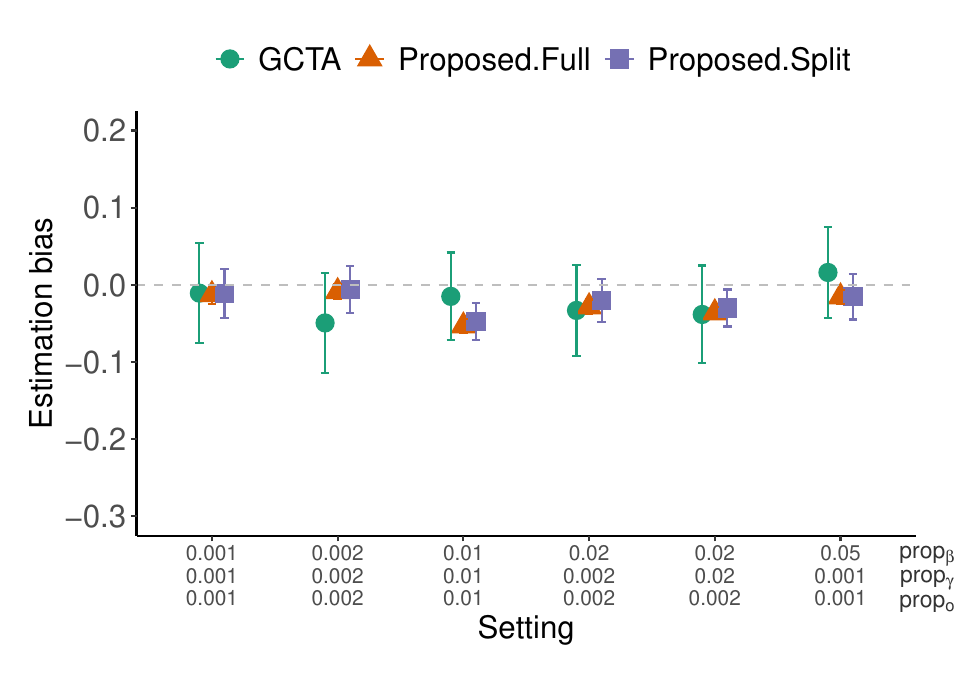} &  
  \includegraphics[height=0.30\textheight,width=0.5\textwidth]{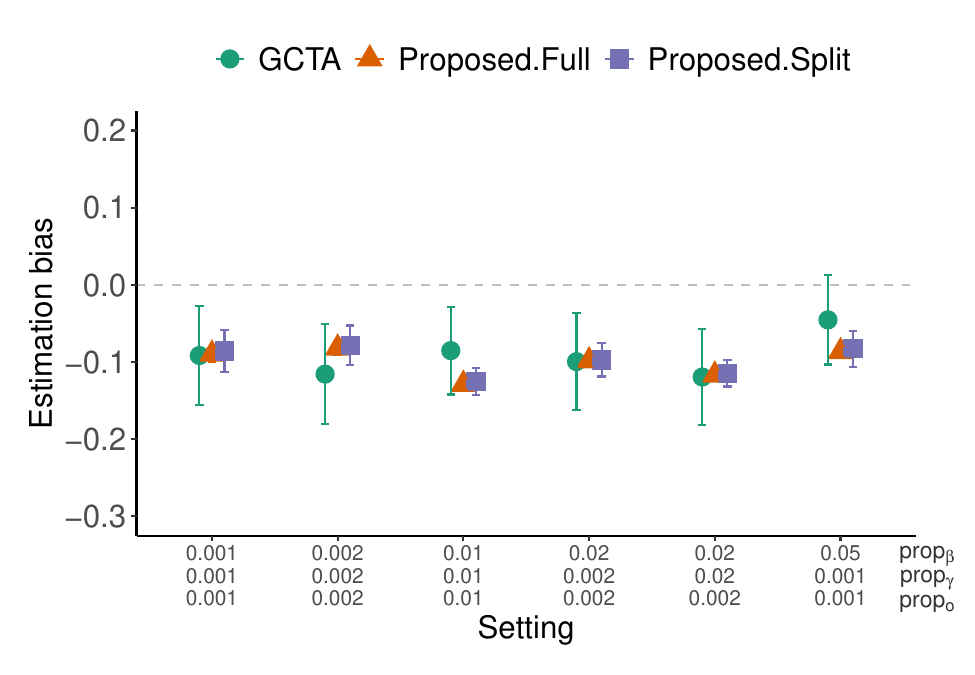} \\
  (e) Single mis-specified dominant model  & (f) Two mis-specified dominant models \\
\end{tabular}
}
\caption{Comparison of estimation accuracy for different genetic models  when  the two traits are measured on two independent  set of individuals ($N_y = N_z = 8000$).  Each point represents the mean of the average bias, with error bars representing 1.96 times the standard error on both sides. }
\label{fig:LL-noover}
\end{figure}
}

\noindent {\bf Estimation errors}: The results of estimation accuracy when both traits are measured on the same set of individuals (overlapping samples) are present in Figure \ref{fig:LL-over}.
 Under correct linear model specifications with different effect distributions, Figures \ref{fig:LL-over} (a), (b) and (c) show  the proposed methods provide approximately unbiased estimates of genetic covariance when two models are not highly polygenic, not sensitive to the underlying effect distribution.
 The full-sample estimator and split-sample estimator are almost unbiased  when both traits have sparse or moderate polygenic signals. Even if one trait is highly polygenic $(\prop = 0.1)$, both estimators still provide satisfactory estimates of  the genetic covariance. In addition, when both traits are highly polygenic $(\prop_\beta, \prop_\gamma) = (0.1, 0.1)$, the estimation is still accurate if the genetic covariance is due to  the shared variants with large effect sizes $(\prop_o = 0.0025).$ As a comparison, the estimates from  GCTA could be biased downward (Figure \ref{fig:LL-over} (b) ) or upward (Figure \ref{fig:LL-over} (c)) when the distributions of effect sizes  depend on MAF and LD.

  The full-sample estimator and the split-sample estimator have their own strengths and limitations. Consistent with the theoretical analysis,  the split-sample estimator has a smaller bias under unbalanced sparsity settings, i.e. one trait is highly polygenic while the other trait is sparse with $(\prop_\beta, \prop_\gamma, \prop_o) = (0.1, 0.0025, 0.0025)$. We provide more results in Figure S6 of the supplementary material to highlight this point. 
 On the other hand, the full-sample estimator has better empirical performance when the genetic covariance is contributed by weak and dense effects $(\prop_\beta, \prop_\gamma, \prop_o) = (0.05, 0.05, 0.05)$. In Section C of the Supplemental Materials,  we further investigate the different behaviors using oracle estimators in linear models, which suggests that the difference in performance is related to the correlation between the error terms in the trait models. 


Figure \ref{fig:LL-over} (d) shows both methods  give accurate estimations of the true genetic covariance  under single model mis-specification, which is explained by Corollary \ref{coroll:narrow-sense} that  the narrow-sense genetic covariance is equal to the true genetic covariance  when only one model is mis-specified.


\ignore{When both models are mis-specified, both methods also have similar bias. When both models are mis-specified, the bias is explained by the difference between the narrow-sense genetic covariance and true genetic covariance. 
}

\noindent 
{\bf Coverage of confidence intervals}: 
Furthermore, in Figure  \ref{fig:LL-over-debias},
we compare the coverage probabilities and the length of the confidence interval. When the distributional assumption is violated,  the coverage probabilities of GCTA are low due to their biased estimation. For the split-sample method,  the coverage probabilities of the confidence intervals are close to the nominal level when at least one trait has sparse signals. The results  support our argument that the proposed method does not require both models to have  sparse coefficients for valid inference.

	The full-sample approach has  coverage probabilities close to the nominal level  when both trait models are sparse, while its  performance gets worse when the bias dominates over the variance term in the unbalanced sparsity or dense effects setting.     
 Therefore we provide  a bias-adjusted confidence interval estimation based on theoretical results \citep{celentano2020lasso,tibshirani2012degrees}  and evaluate its performance.  We quantify the bias  by $\hat R$ based on the prediction risk and residuals, and propose the following  adjusted confidence interval, 	
	\begin{align}
		\label{adjust-CI}
		CI(\alpha) =  \begin{cases} [\widehat I_{\text{full}} - z_{1 - \alpha/2} \widehat \sigma - \widehat R , \widehat I_{\text{full}} + z_{1 - \alpha/2} \widehat \sigma ] & \text{ if } \widehat R > 0\\
			[\widehat I_{\text{full}} - z_{1 - \alpha/2} \widehat \sigma , \widehat I_{\text{full}} + z_{1 - \alpha/2} \widehat \sigma -\widehat R ] &  \text{ if } \widehat R < 0
		\end{cases}.
	\end{align}
	More details on the derivation of $\hat R$ can be found in Section C of the Supplemental Materials.  Figure  \ref{fig:LL-over-debias} shows the adjusted CI leads to valid coverage even when the effects are highly polygenic, while maintaining a shorter width than the estimate from  GCTA method in most cases.

Finally, if one model is mis-specified,   our  proposed method still performs well.     These results support our conclusion that the inference of genetic covariance between the linear trait and the nonlinear trait is robust to the mis-specification of the nonlinear function.

\begin{figure}
	\ignore{
	\resizebox{0.95\linewidth}{!}{
		\begin{tabular}{cc}
			\includegraphics[height=0.3\textheight, width=0.5\textwidth]{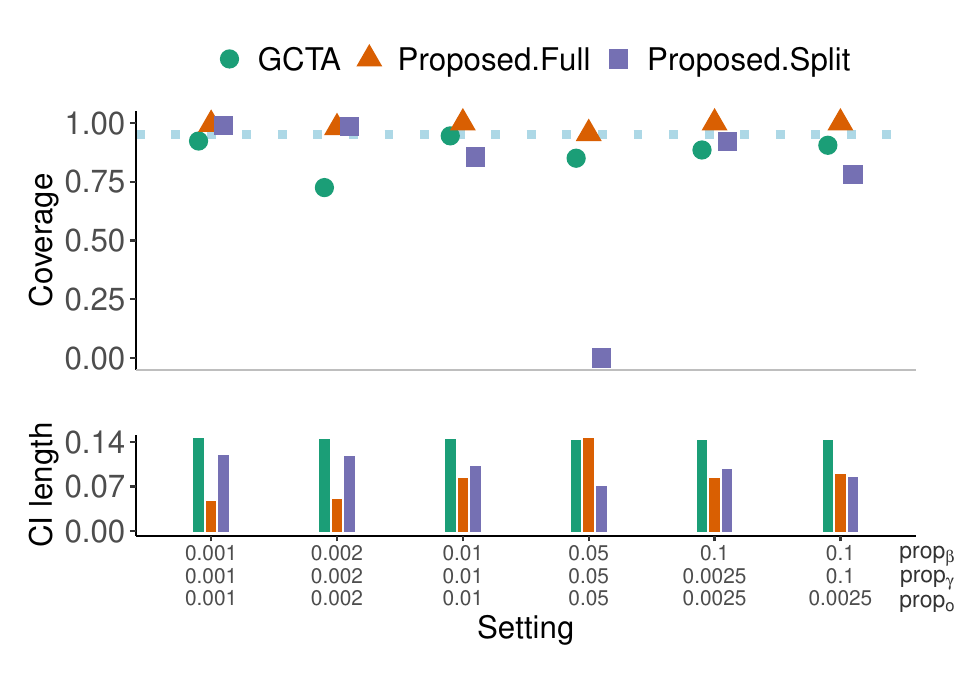} &  
			\includegraphics[height=0.3\textheight,width=0.5\textwidth]{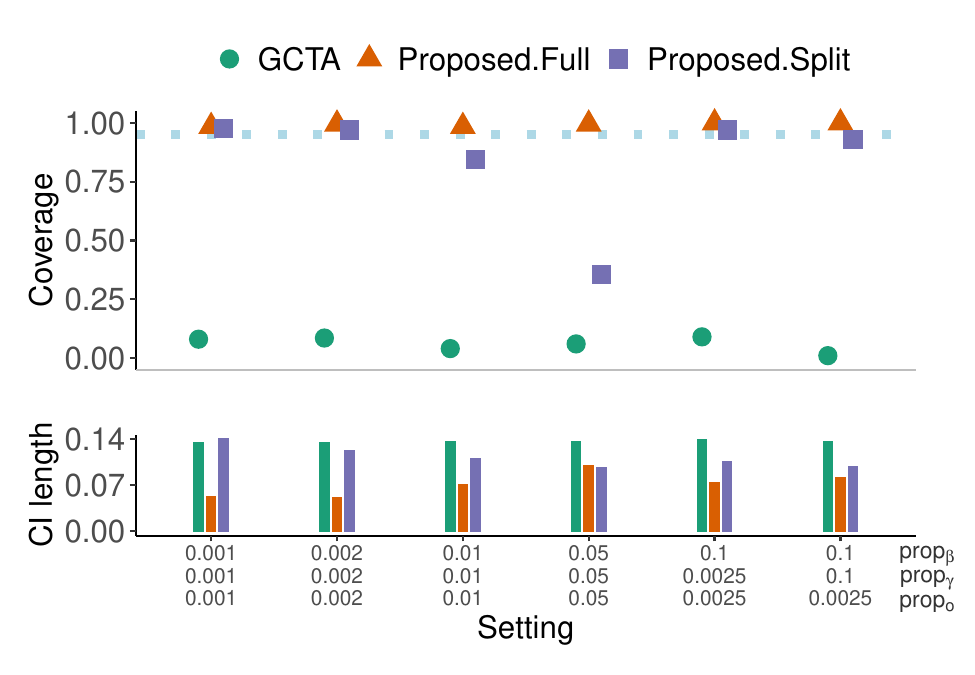} \\
		\vspace{-0.1in}
			(a) Normal coefficients, linear models & (b) LDAK 1, linear models \\
			\includegraphics[height=0.30\textheight, width=0.5\textwidth]{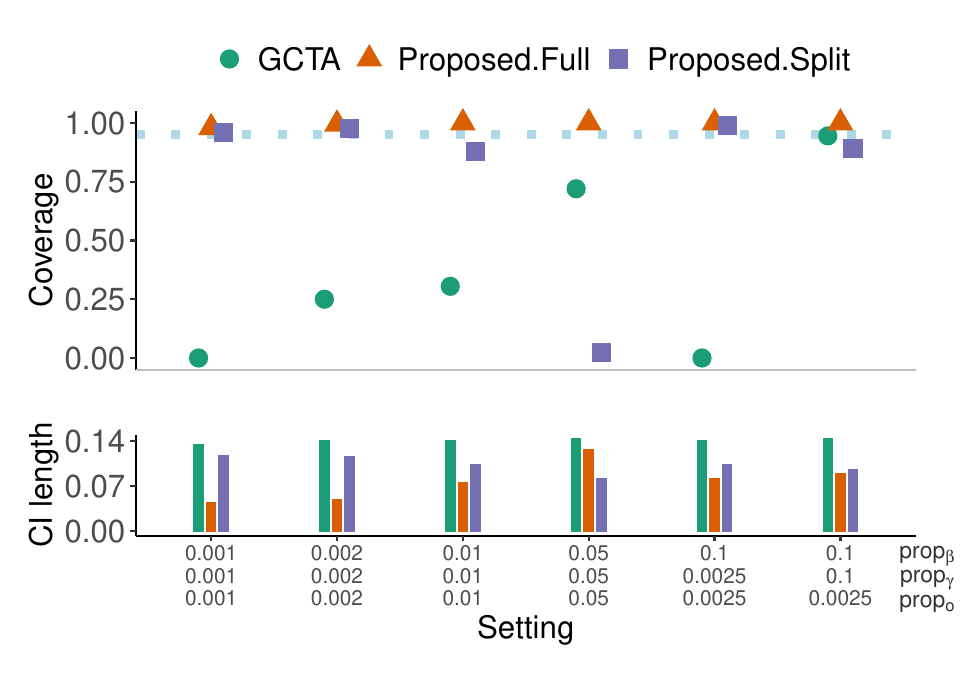} &  	\includegraphics[height=0.30\textheight,width=0.5\textwidth]{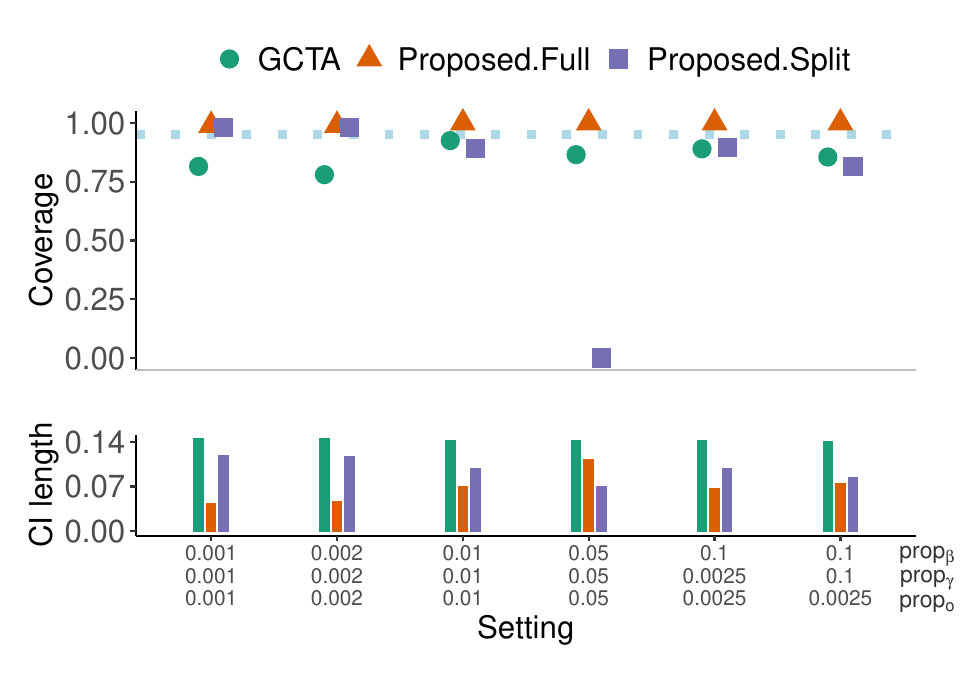} \\
		\vspace{-0.1in}
				(c) LDAK 2, linear models &  (d) Single mis-specified composite model \\
	\ignore{				\includegraphics[height=0.30\textheight,width=0.5\textwidth]{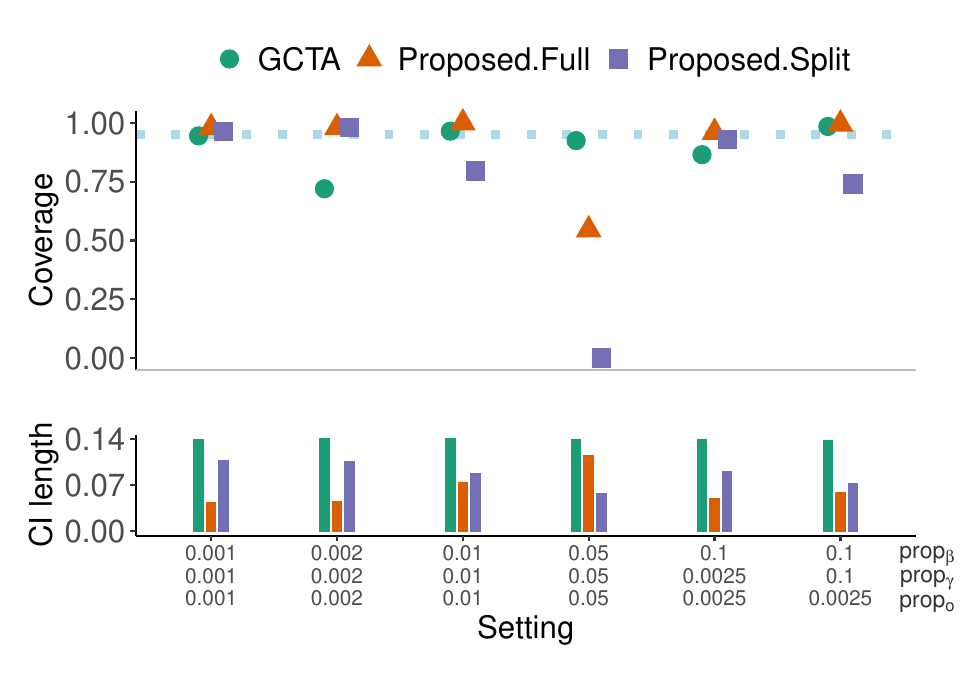} & \\
	\vspace{-0.1in}
				(e) Single mis-specified dominant model & \\
		}
		\end{tabular}
	}
}
\resizebox{0.95\linewidth}{!}{
	\begin{tabular}{cc}
		\includegraphics[height=0.3\textheight, width=0.45\textwidth]{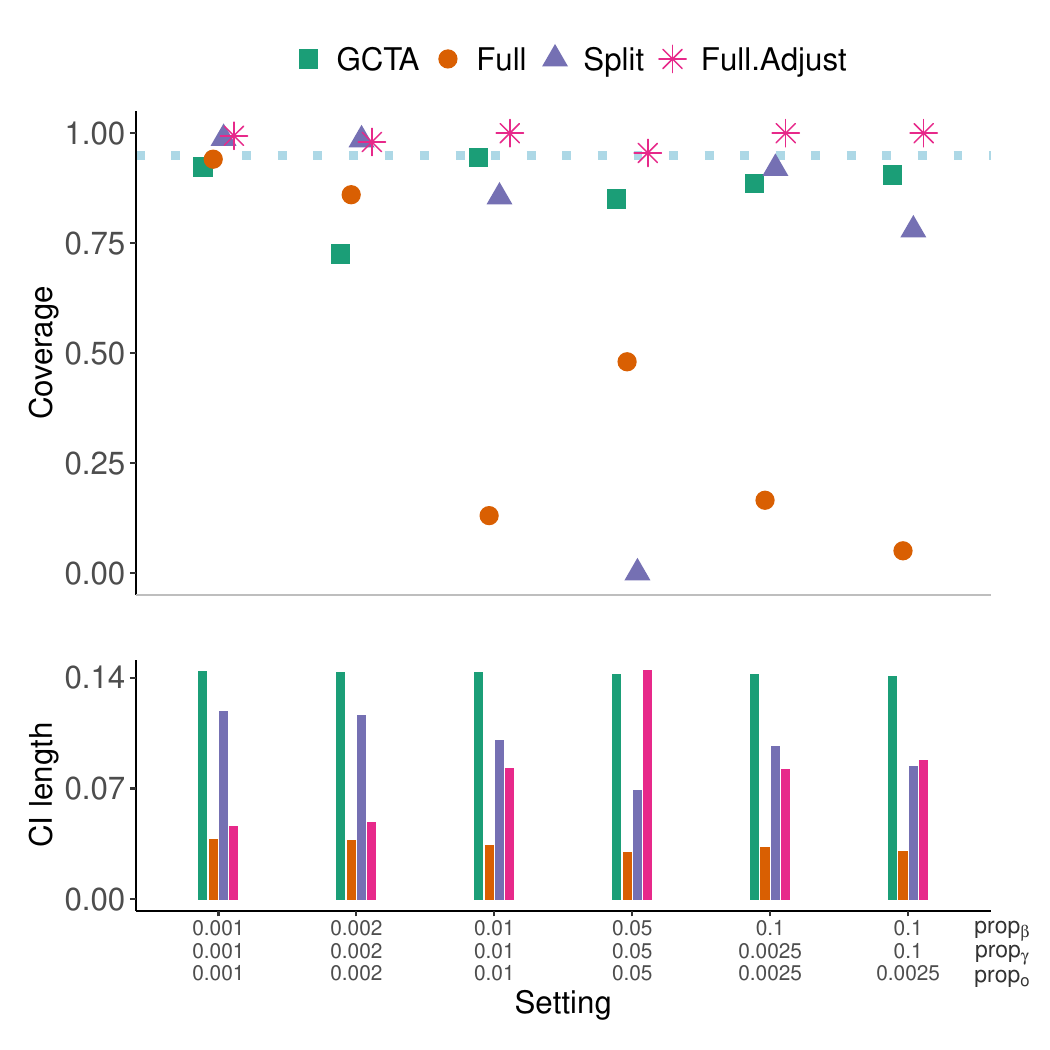} &  
		\includegraphics[height=0.3\textheight,width=0.45\textwidth]{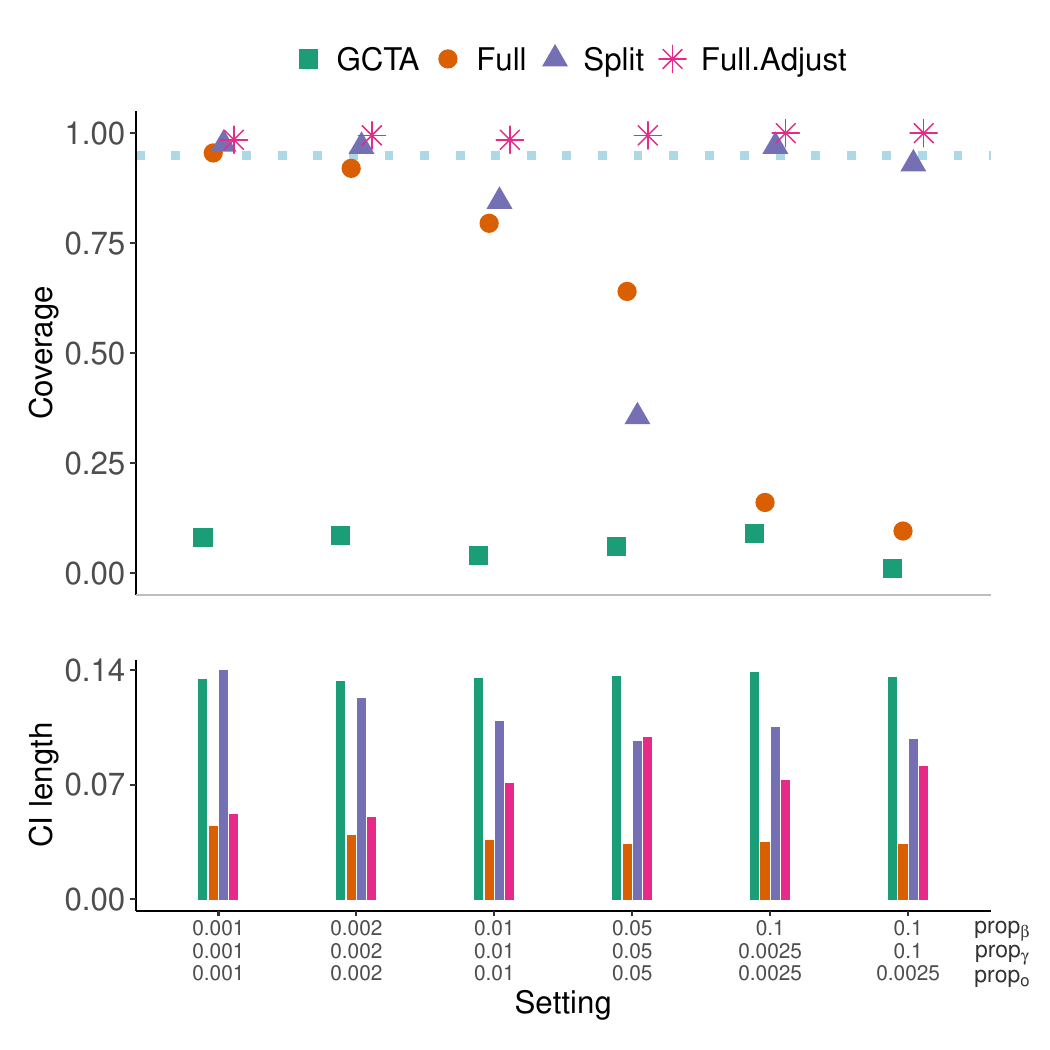}  \vspace{-0.1in}\\
		\vspace{-0.02in}
		(a) Normal coefficients, linear models & (b) LDAK 1, linear models \\
		\includegraphics[height=0.30\textheight, width=0.5\textwidth]{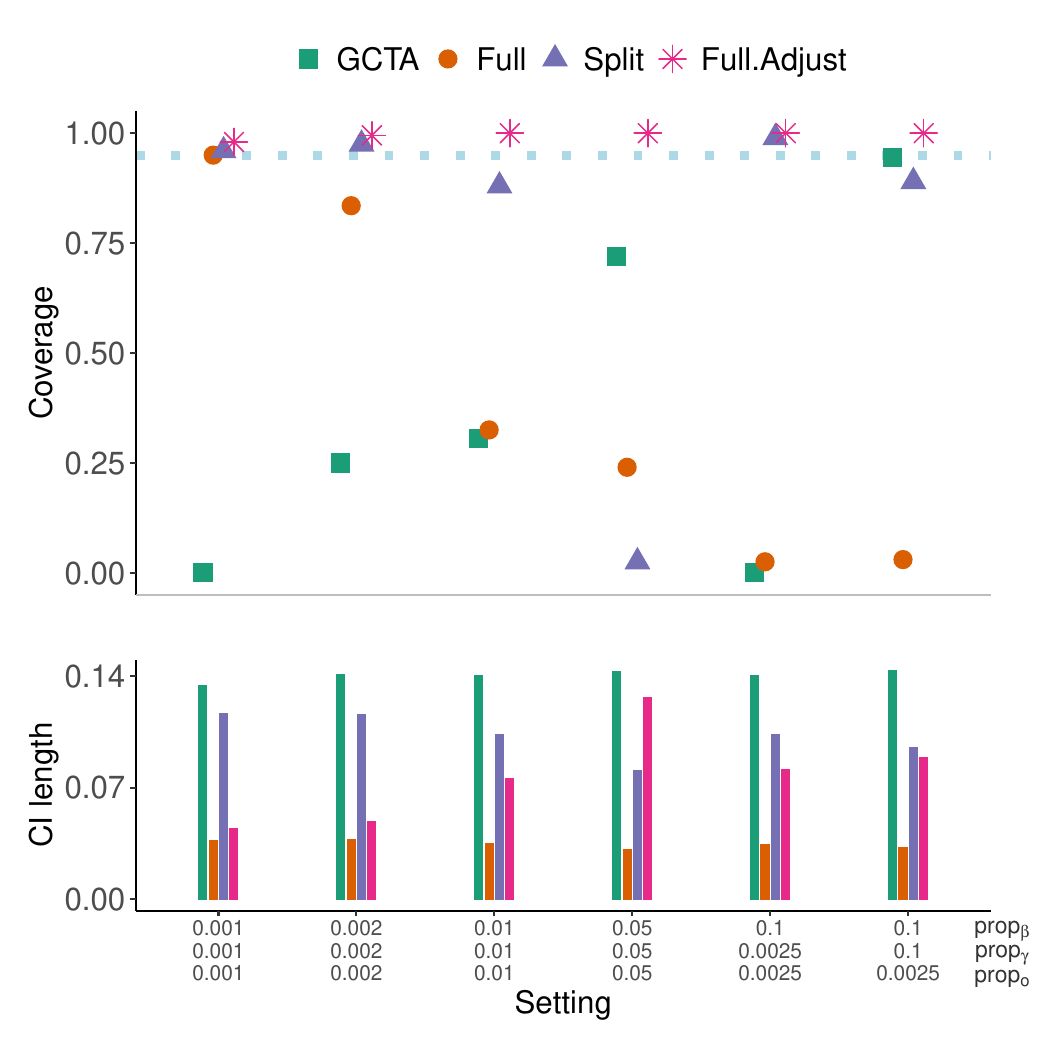} &  	\includegraphics[height=0.30\textheight,width=0.5\textwidth]{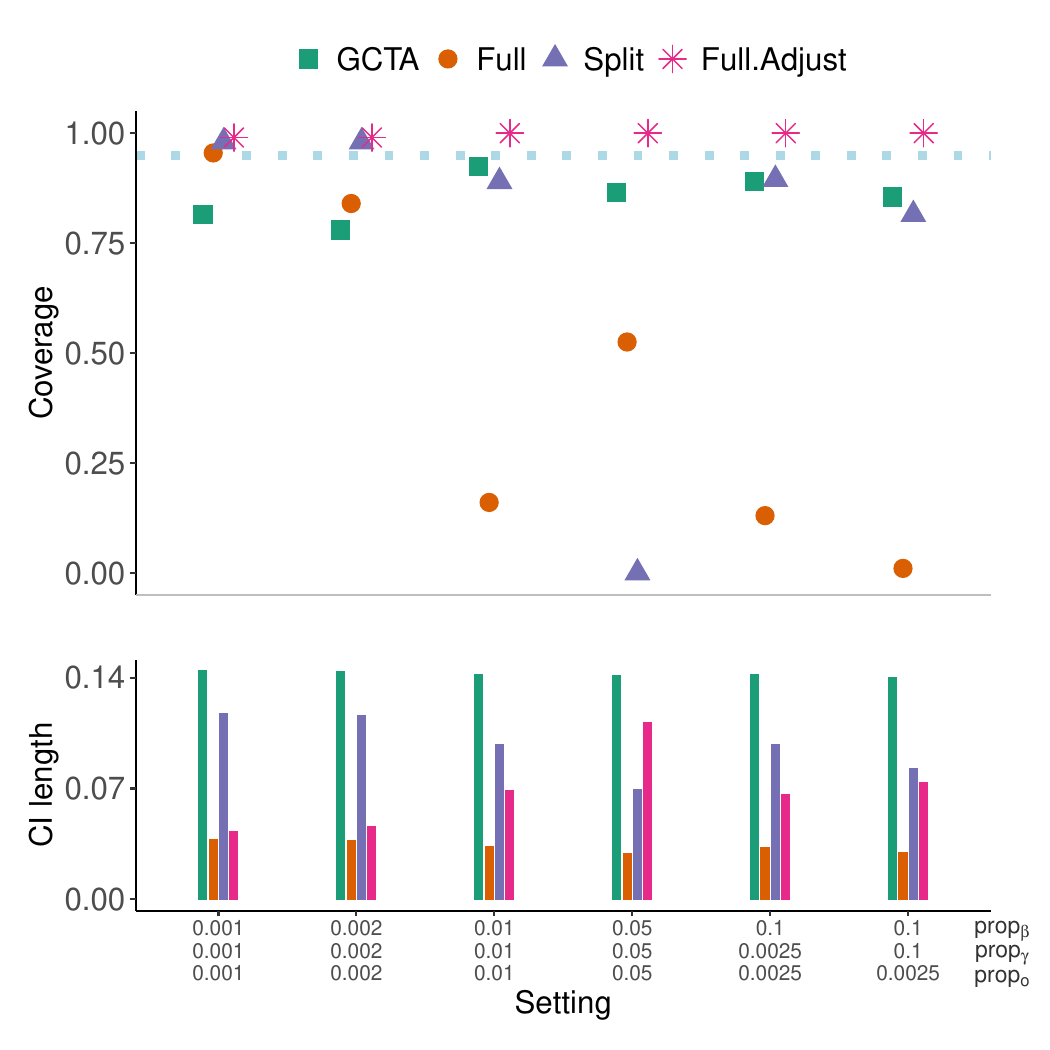}  \vspace{-0.1in}\\
		\vspace{-0.02in}
		(c) LDAK 2, linear models &  (d) Single mis-specified composite model \\
		\ignore{				\includegraphics[height=0.30\textheight,width=0.5\textwidth]{./Figure/Simu/result2023/debias-3.pdf} & \\
			\vspace{-0.1in}
			(e) Single mis-specified dominant model & \\
		}
	\end{tabular}
}

	\caption{Comparison of coverage probability and confidence interval (CI) length when both traits are measured on the same set of individuals ($N_y = N_z = 8000$). Conservative CI is constructed for the full-sample approach. The coverage probability results are reported in the upper panel and the corresponding average confidence interval length results are demonstrated by the bar plots in lower panel.   
	}
	\label{fig:LL-over-debias}
\end{figure}

\noindent {\bf Non-overlapping setting}:  Results  when the two traits are measured on two independent sets of individuals  are given  in  Figure S1  and S2 in the Supplemental Materials. The split-sample and full-sample estimators have similar estimation biases and the bias is smaller when the effects are sparse (Figure S1). However, when the models are very polygenic, the full-sample approach has low coverage probability due to a smaller variance. In contrast, the split-sample approach gives better coverage (Figure S2). Similar estimation behaviors of two methods are observed under model mis-specification, and the proposed method is robust to model mis-specification (Figure S2 (b)-(d)).

\subsection{Summary of additional simulations}
In Supplemental materials, we include simulation results for binary traits and case-control designs (Figures S3-S4, Table S2). The results  show the proposed method can still estimate the genetic covariance well.  When the models are not highly polygenic, the coverage probabilities in Figure  S4 are also close to the nominal level.  In contrast,  GCTA  seems to have a large estimation bias for the genetic covariance even under the models with normal coefficients (Figure  S3 (a)), resulting in  low  coverage probabilities  in some cases. When one model is misspecified,  GCTA estimation results in a larger bias (Figure S3 (d)). This may be  due to the  wrong working models fitted by  GCTA-GREML for the binary traits. 

We also evaluate the performance when causal variants are not independent in Figure S5. The proposed methods still perform well. When  compared  with the models with independent causal variants as shown in Figure \ref{fig:LL-over}, the estimates from   GCTA have a larger bias when the signals are sparse, leading to lower corresponding coverage probabilities.  

\section{Real data application}
\label{sec:real-data}

We analyze the outbred Carworth Farms  White (CFW)  mice data set \citep{parker2016genome} to study the  genetic covariance between various behavioral and physiological traits.  
Each CFW mouse was phenotyped for behavioral traits, bone and muscle traits, and other physiological traits, including  fasting glucose levels, body weight, tail length and testis weight. The behavioral traits consist of conditioned fear, methamphetamine (MA) sensitivity and prepulse inhibition phenotypes. The bone or muscle traits include the weight of five hindlimb muscles and bone mineral density. Besides, a binary trait that signals abnormally high bone mineral density is generated. For each of the continuous phenotypes,  we adjust for baseline weight, experimenters, sacrifice age and 1st PC of genotypes and normalize the residuals.  After the pre-processing, the data set consists of 1038 mice with 79,824 genetic variants (SNPs), 66 continuous phenotypes and one binary trait.  The data set includes various levels of missingness $(22 \sim 211)$ in the trait values. The observations with missing traits are not used in the GLM fitting, but they are used in estimating the genetic covariance using the imputed trait values. 

\ignore{

The behavioral traits were measured through a series of experiments \citep{parker2016genome}. For methamphetamine sensitivity traits, the locomotor activity of each mouse was measured by total distance traveled  (in cm) over the 30-minute interval immediately after the injection in each day. The activity and time spent in the center of the arena on the first day were also recorded. The mice were injected with saline on first two days and received MA on the third day. 
The conditioned fear traits measure the recall of the fearful memory by measuring the mice's freezing behavior in response to the stimulus. After the mice being placed in the test chambers,  the average proportion of time freezing over the 30-180 seconds interval is recorded at each day. Besides, binary traits were also created for fear conditioning metrics that show a "bimodal" freezing distribution. Finally, the 
prepulse inhibition (PPI) is a reduction in startle response that occurs when a non-startling lead stimulus (“prepulse”) precedes a startling stimulus. The  PPI traits were defined as the difference of the average startle amplitude during the 3, 6 or 12-db prepulse trials to the average startle amplitude during the pulse-alone trials.
}

The phenotypic correlations among the continuous traits  are shown   in Figure \ref{fig:pheno-cor} (a). The heritability estimates of continuous traits  using GCTA and our proposed method are present in Figure \ref{fig:pheno-cor} (b).  Due to the limited sample size, the proposed method may not be able to capture the signal when traits have dense and weak effects. 
For the subsequent analysis, we consider 22 continuous traits with estimated heritability larger than 0.10 by the proposed method.

\begin{figure}[hbt!]
	\resizebox{\linewidth}{!}{
	\begin{tabular}{cc}
		\includegraphics[height=0.25\textheight, width=0.5\textwidth]{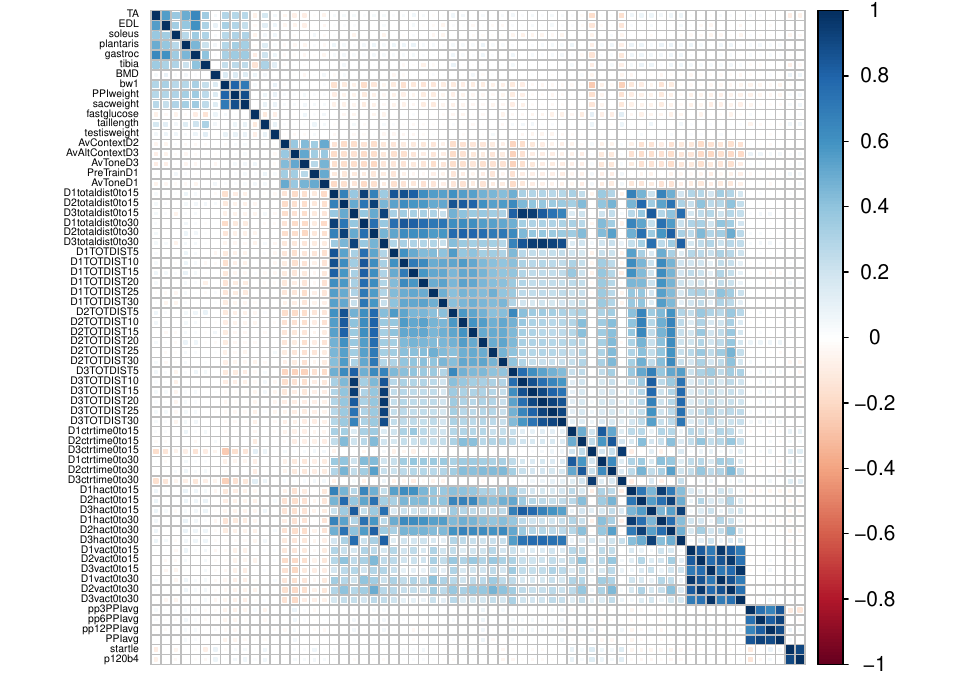} &
		\includegraphics[height=0.25\textheight, width=0.5\textwidth]{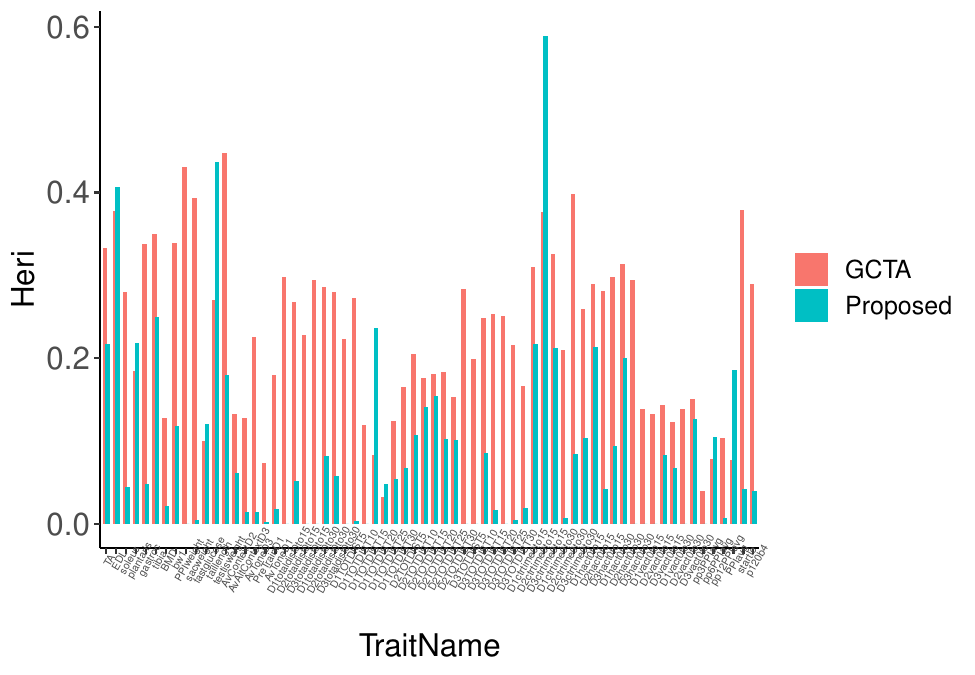} \\
		(a) Phenotype correlations among different traits. & (b) Estimated heritability of different traits. \\
	\includegraphics[height=0.25\textheight, width=0.5\textwidth]{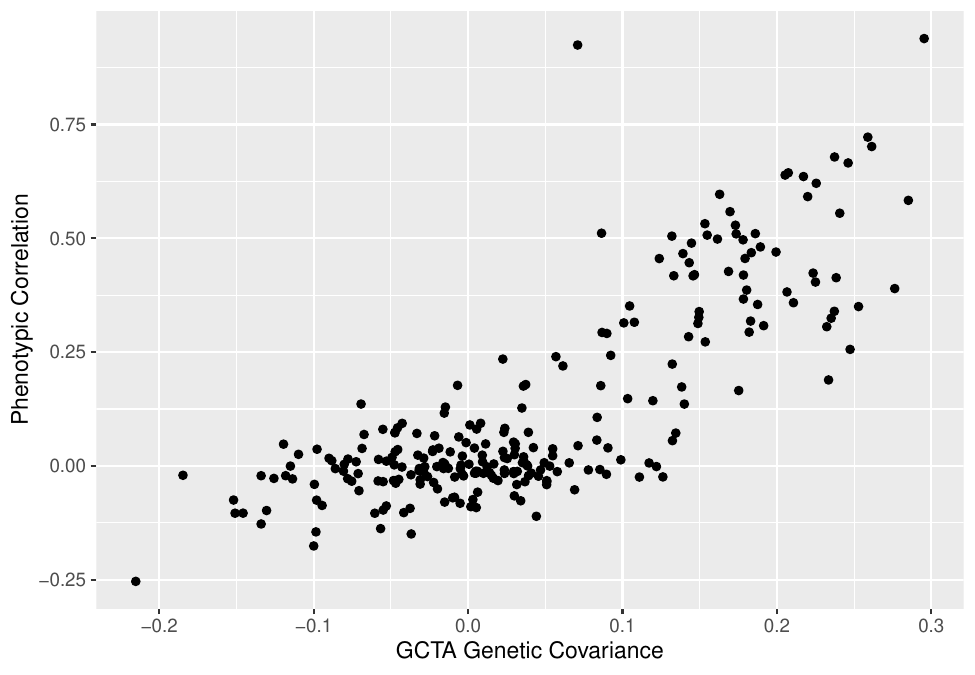} &
	\includegraphics[height=0.25\textheight, width=0.5\textwidth]{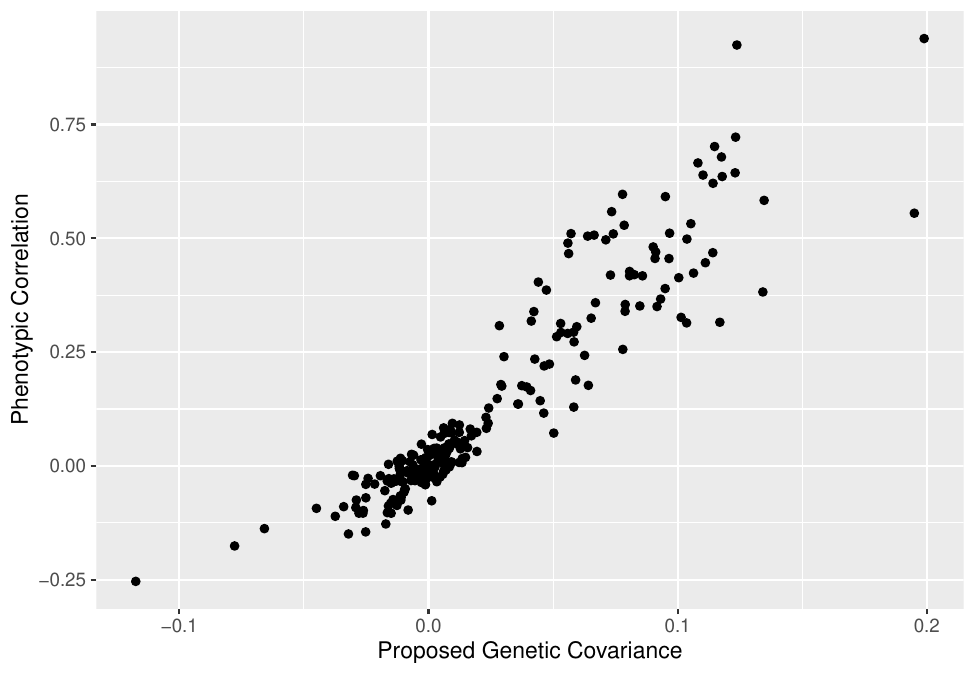} \\
	(c) GCTA estimator  &
	(d) Proposed estimator  \\
			\includegraphics[height=0.25\textheight, width=0.5\textwidth]{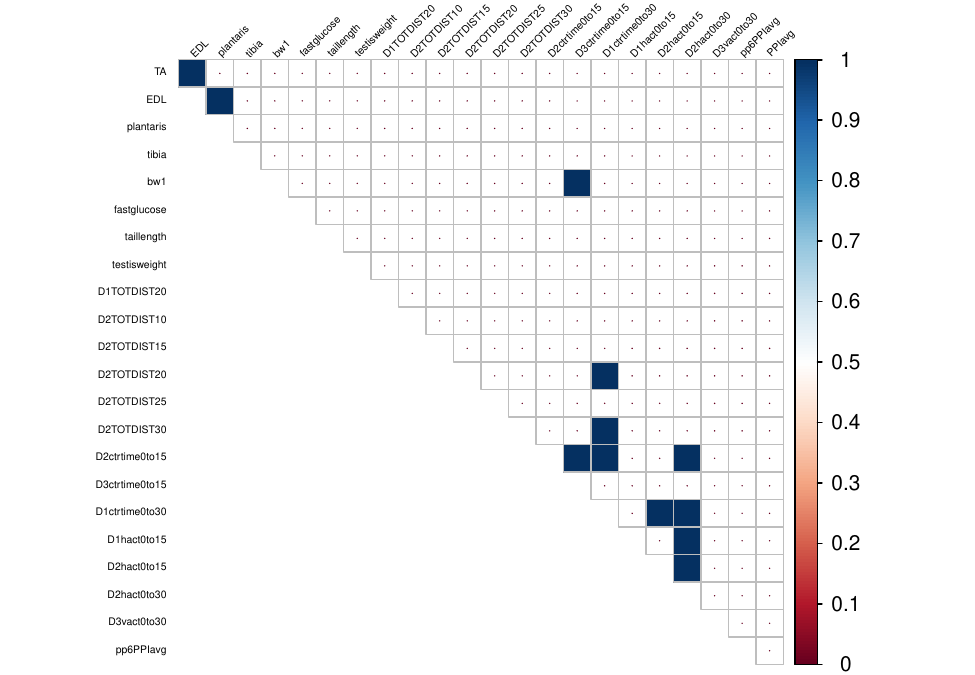} &
	\includegraphics[height=0.25\textheight, width=0.5\textwidth]{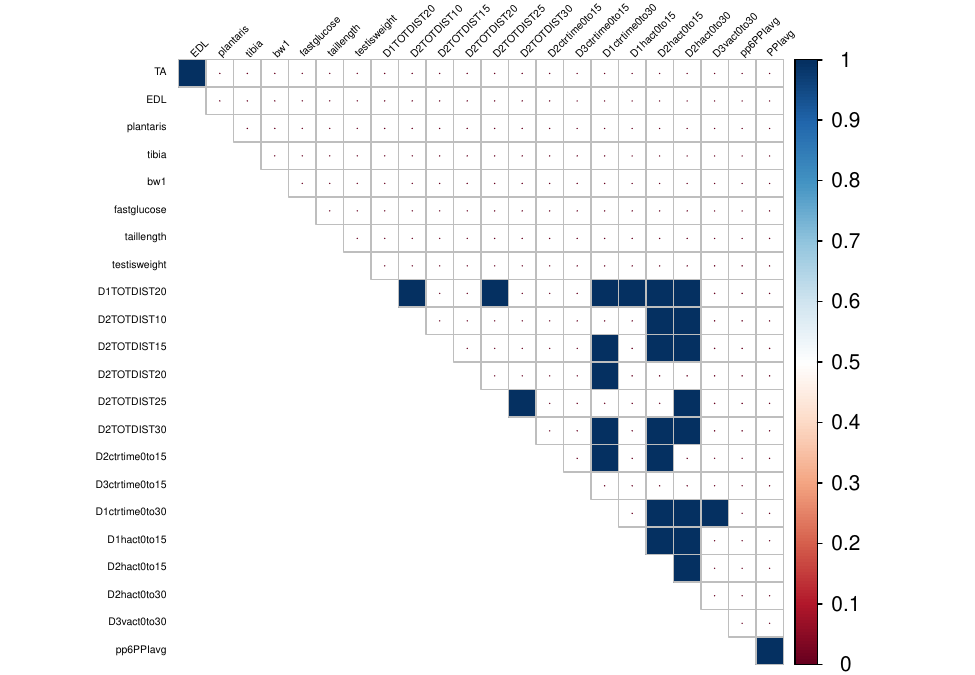} \\
	(e) GCTA estimator (Bonferroni)  &
	(f) Proposed estimator (Bonferroni)\\
\ignore{	  \includegraphics[height=0.25\textheight, width=0.5\textwidth]{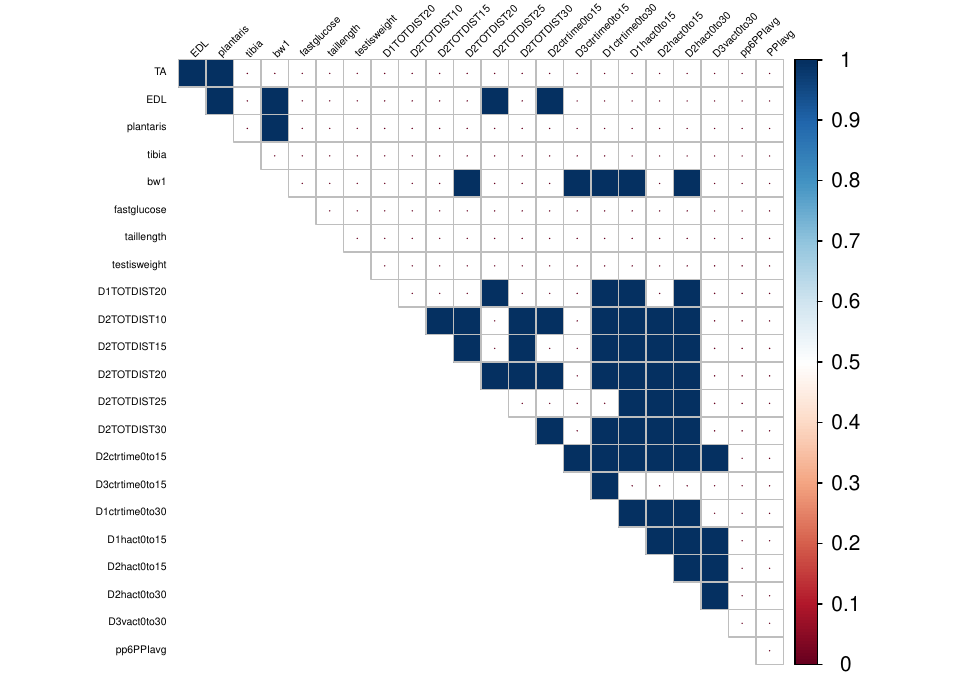} &
	\includegraphics[height=0.25\textheight, width=0.5\textwidth]{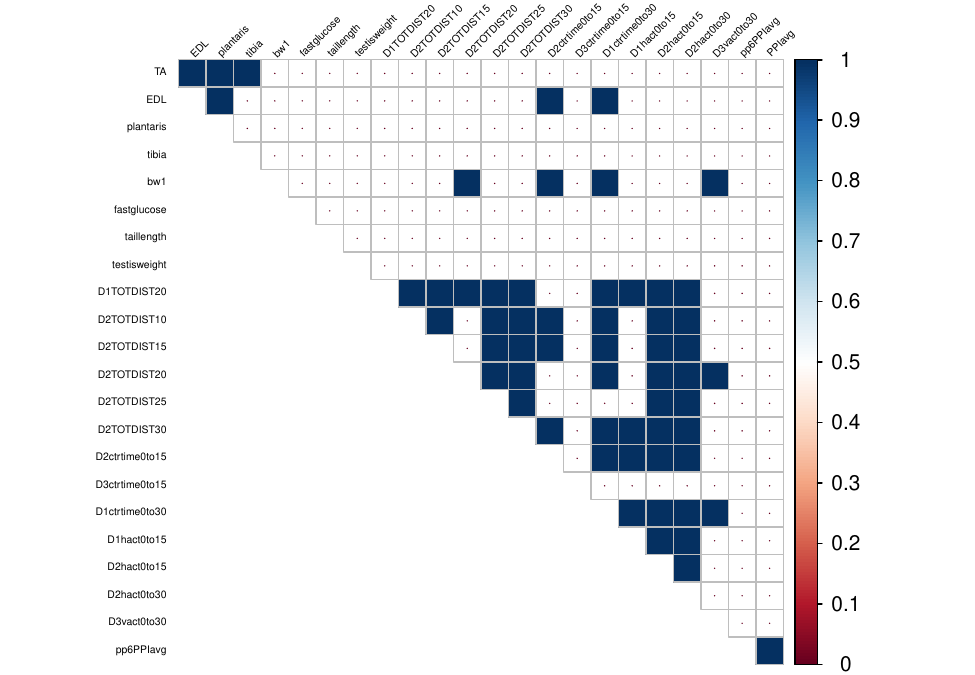} \\
	(c) GCTA estimator (BH)  &
	(d) Proposed estimator (BH)  
}
\end{tabular}
}	
	\caption{Results from analysis of the mouse GWAS data. (a) Phenotype correlations of 66 continuous  traits;   (b) Estimated heritability  by GCTA method and the proposed method;
	(c)-(d)	Scatter plots of phenotypic covariance and the estimated genetic covariance using GCTA  (\ and the proposed estimator  for each pair of continuous traits with the estimated heritability larger than $0.1$;
(e) -(f) Identified significant pair of traits with Bonferroni correction at the level of $0.05$ using GCTA estimator  and the proposed estimator.  There are 12 pairs identified by  GCTA and 27 pairs identified using the proposed method.	
}\label{fig:pheno-cor}
\end{figure} 

\subsection{Genetic covariance among different traits}
To illustrate our method,  we present the genetic covariance analysis for four selected traits. For the muscle traits, we choose the continuous trait extensor digitorum longus (EDL) and the binary trait signaling abnormal bone.  Other traits  include the physiological trait of testis weight,   the MA sensitivity trait of  the distance traveled, 0–30 min, on day 3 of methamphetamine sensitivity tests. For each selected trait, we calculate its genetic covariance with all other traits.
For the binary trait, the estimated genetic covariance is calculated based on the observed scale.
Model fitting and estimation of $\hat \beta$ and $\hat \gamma$ are implemented using the package \textbf{glmnet} with 10-fold cross-validation. 

The genetic covariance results for the four traits are summarized in Figure \ref{fig:all-gecv} (a) - (d). For continuous traits, we report the bias-adjusted confidence interval in \eqref{adjust-CI}.
Figure \ref{fig:all-gecv} (a) shows that the testis weight has no significant genetic covariance with other continuous traits. Figure \ref{fig:all-gecv} (b) suggests that the abnormal bone trait may be mainly related to the muscle traits. The scale of the genetic covariance is small because the abnormal bone trait is binary.
In Figure \ref{fig:all-gecv} (c), the muscle trait EDL has a significant genetic covariance with other muscle traits like TA but does not have a shared genetic architecture with behavior traits. This agreed with the  pleiotropy effects among the muscle traits \citep{parker2016genome}.
In Figure \ref{fig:all-gecv} (d), the MA sensitivity trait is closely related to other MA sensitivity behavior traits measure on Day 1 and 2 and is not related with trait measured on the third day. This result is reasonable because the mice were injected with saline on first two days and received MA on the third day so the activity on the day 3 should be less related to the baseline measurement on day 1 and 2.  

These results  show that our proposed estimator of genetic covariance works well for both continuous and binary traits. The results further confirm that many behavior-related traits share common genetic variants and physiological traits also share genetic effects. However,  the genetic covariance between physiological traits and behavioral traits is small.

\begin{figure}[hbt!]
	\resizebox{\linewidth}{!}{
		
\begin{tabular}{cc} 
\includegraphics[height=0.30\textheight, width=0.50\textwidth]{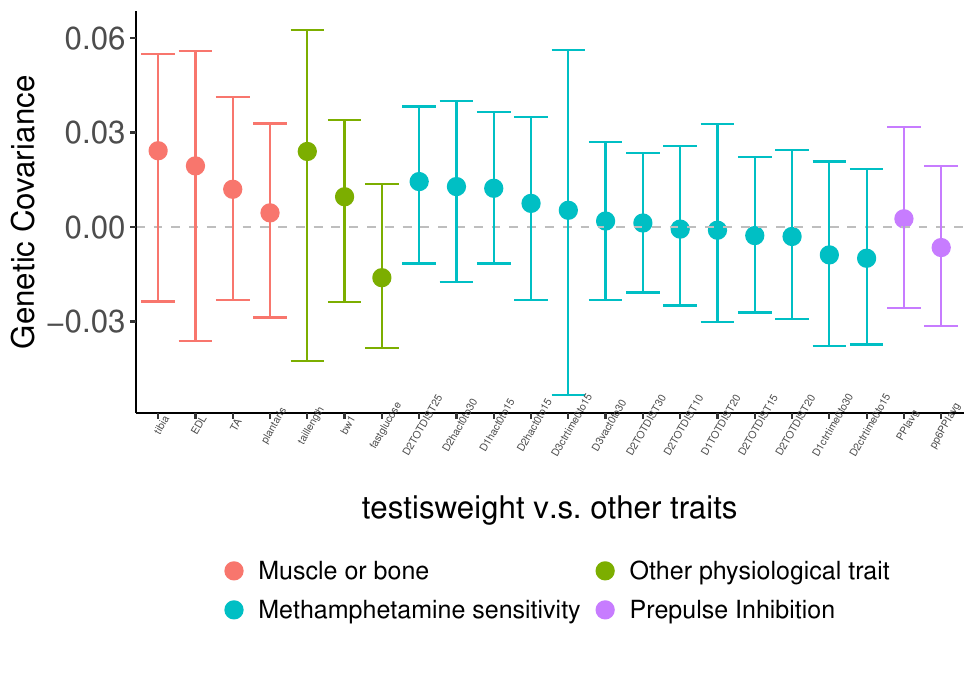}
&
\includegraphics[height=0.30\textheight, width=0.50\textwidth]{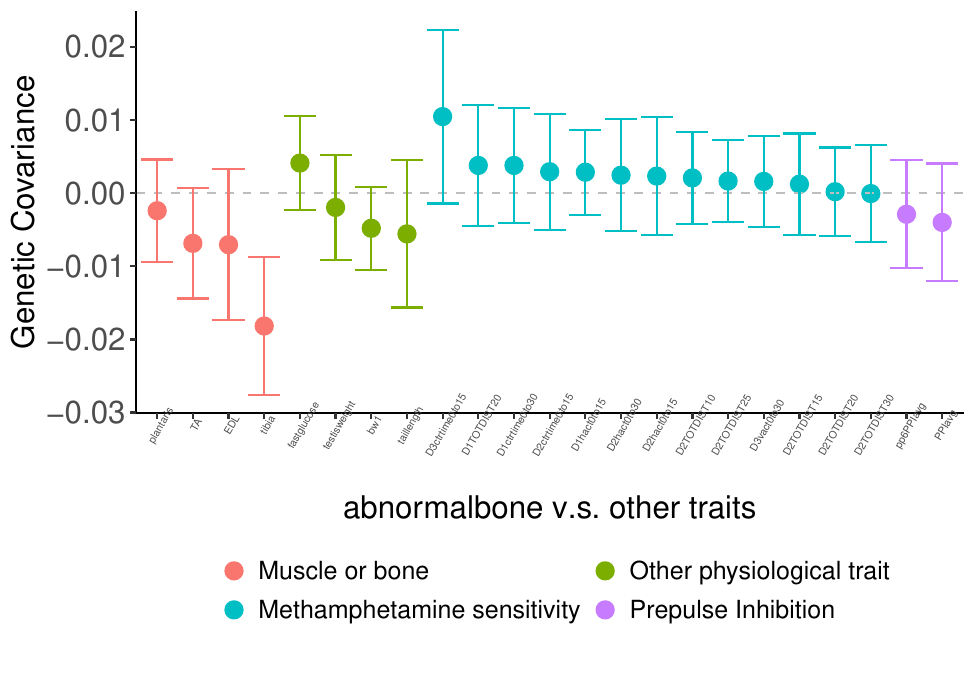}\\
(a) Physiological trait,  testis weight.
&
(b) Binary bone trait, abnormal bone.
\\
\includegraphics[height=0.30\textheight, width=0.50\textwidth]{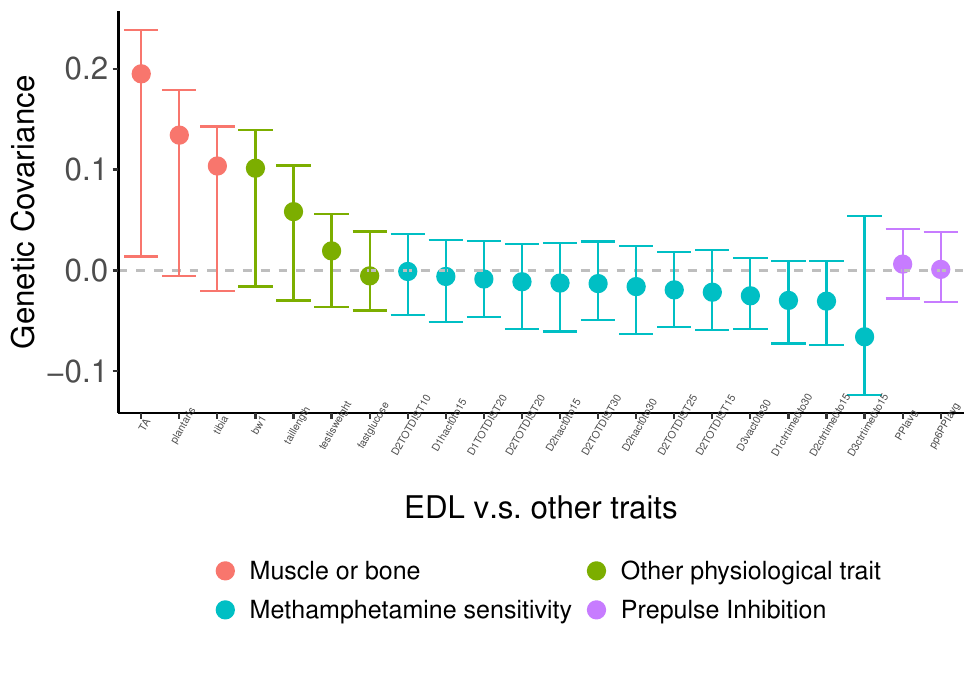}
&\includegraphics[height=0.30\textheight, width=0.50\textwidth]{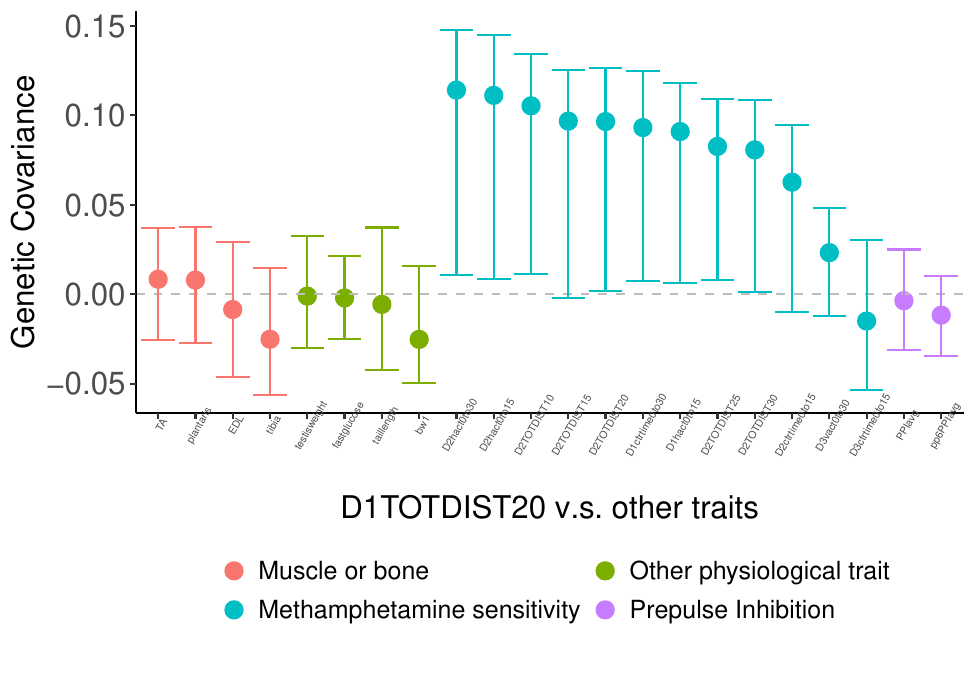}
\\
{(c) Muscle trait} &{(d) Methamphetamine sensitivity trait}\\
weight of extensor digitorum longus muscle.
& the distance traveled, 0–20 min. 
\end{tabular}
}
\caption{Estimated genetic covariance of four selected traits  with other traits, where  bars represent the confidence intervals. For continuous traits, we report the conservative adjusted confidence interval. The confidence level $\alpha$ is determined by Bonferroni  correction $0.05/21$.}
\label{fig:all-gecv}
\end{figure}

%

\subsection{Comparison with GCTA and LDSC estimates}

We compare the estimates of genetic covariance from the  proposed method, GCTA and the LD-score regression.    The scatter plots in Figure \ref{fig:pheno-cor} (c)-(d)  show  the relationship between estimated genetic covariance and the observed phenotypic covariance. When  the phenotypic correlation is around $0$, the proposed method gives  estimates of the genetic covariance around $0.$ However, the estimates from   GCTA  could be from $-0.2$ to $0.15$, leading to inflated estimates.  GCTA and LDSC estimations for four selected traits above are given in Figure S7 and S8.
Additional  comparisons of genetic covariance estimates by three methods are given in Figure S9.  

We also compare the performances for identifying the non-zero genetic covariance based on the  confidence interval estimation with  Bonferroni correction for multiple comparisons (see Figure 
\ref{fig:pheno-cor}). The bias-adjusted confidence interval is used for the proposed method.  Our method has identified more  trait pairs (27 pairs) than GCTA (12 pairs) with a significant genetic covariance.  
The proposed method identifies more pairs related to the traveling distance traits measured in MA sensitivity experiments.   Finally, as shown by Figure S10, the LDSC method only detect 6 pairs due to the fact that it uses summary statistics only.



\section{Discussion}
\label{sec:discuss}

This paper proposes a general regression-based estimation and inference procedure for the genetic covariance and the  narrow-sense genetic covariance based on GWAS data. The proposed estimator enjoys asymptotic normality under proper conditions and is robust to model mis-specifications. Numerical studies are conducted to explore its empirical performance under various sparsity levels. 
	The accuracy of the proposed regression-based method depends on the precision of the fitted model and it works best under uncertain sparsity conditions, outperforming the random-effects model-based estimation. The proposed method may underestimate the true genetic covariance when it is contributed by weak and dense effects, in which case the fitted model would be inaccurate for limited sample sizes.  	
In numerical experiments, we also find the robustness of the full-sample estimator outside  the sparsity regime. Further theoretical investigation would be interesting. 

The proposed estimator $\Iss$ in (\ref{eq-Ihat}) only uses the predicted trait values and residuals. We consider the $\ell_1$-penalized regression in (\ref{eq:def-hat-beta-gamma}) and (\ref{eq:def-hat-beta-gamma2}) for its convenience in prediction and estimation, and its  theoretical guarantees under the sparsity assumptions.  It is however,  possible to apply other appropriate penalties or more flexible machine learning methods in the first step. Such machine learning approaches can be theoretically justified if at least one model provides accurate predictions as studied in the double machine learning literature \citep{chernozhukov2018double}.

 The practical implementation of the proposed methods for large-scale genotype data involves an  efficient implementation of penalized GLM regression, which has been shown to be feasible for the UK Biobank genotype data \citep{prive2019efficient,qian2020fast}. Computational tricks such as using marginal association statistics or well-designed iteration rules have been applied to improve efficiency. The tuning parameter selection using  cross-validation  is computationally intensive. Further exploration of other efficient tuning parameter selection methods for  GLMs is warranted.

\if1\blind
{
\section*{FUNDING}
This research was supported by NIH grants R01GM123056 and R01GM129781.
Sai Li's research was also supported by the Fundamental Research Funds for the Central Universities, and the Research Funds of Renmin University of China. 
}\fi

\begin{center}
	{\large\bf SUPPLEMENTARY MATERIAL}
\end{center}
Supplement to ``A Regression-based Approach to Robust Estimation and Inference for Genetic Covariance''.
In the Supplementary Materials, we provide the proofs of theorems, more discussions on estimation bias and imputation,  and more results for numerical experiments and data applications.



\setlength{\bibsep}{0.3ex plus 0.2ex}
\bibliographystyle{chicago}
\bibliography{GRGLM}

 \end{document}